\documentclass[onecolumn, draftcls, 12pt]{IEEEtran}
\usepackage{amsfonts}
\usepackage{graphicx}
\usepackage{mathrsfs}
\usepackage{amssymb,amsmath}
\usepackage[compress]{cite}
\usepackage{bm}
\usepackage[outdir=./]{epstopdf}
\usepackage{algorithmic}
\usepackage{epsfig,subfigure}
\usepackage{color}
\usepackage{multirow}
\usepackage{amsthm}
\usepackage{array}
\usepackage{tabularx}
\usepackage{multicol}
\usepackage{booktabs}
\usepackage{comment}
\usepackage{pdfpages}
\usepackage{cite}
\usepackage{balance}
\usepackage[colorlinks,citecolor=black,linkcolor=black]{hyperref}

\usepackage{url}
\makeatletter
\renewcommand{\@thesubfigure}{\hskip\subfiglabelskip}
\makeatother
\usepackage[linesnumbered, ruled]{algorithm2e}
\newtheorem{proposition}{Proposition}

\makeatletter
\renewcommand{\maketag@@@}[1]{\hbox{\m@th\normalsize\normalfont#1}}%
\makeatother

\begin{document}
\title{{Frequency Diverse Array-enabled RIS-aided  Integrated Sensing and Communication}}
%\title{Frequency Diverse Array Based Integrated Sensing and Communications with RIS: Beamforming Design and Analysis}
\author{Hanyu Yang, Shiqi Gong, Heng Liu, Chengwen Xing, \textit{Member, IEEE}, Nan Zhao, \textit{Senior Member, IEEE}, and Dusit Niyato, \textit{Fellow, IEEE}
\thanks{H. Yang, H. Liu and C. Xing are with the School of Information and Electronics, Beijing Institute of Technology, Beijing, 100081, China (e-mail: hyyangbit@outlook.com; heng\_liu\_bit\_ee@163.com; xingchengwen@gmail.com).}
\thanks{S. Gong is with the School of Cyberspace Science
and Technology, Beijing Institute of Technology, Beijing, 100081, China (email: gsqyx@163.com).}
\thanks{N. Zhao is with the School of Information and Communication Engineering, Dalian University of Technology, Dalian, 116024, China (e-mail: zhaonan@dlut.edu.cn).}
\thanks{D. Niyato is with the College of Computing and Data Science, Nanyang Technological University, Singapore 639798, (e-mail: dniyato@ntu.edu.sg).}}\maketitle

\begin{abstract}
    Integrated sensing and communication (ISAC) has been envisioned as a prospective technology to enable ubiquitous sensing and communications in next-generation wireless networks. In contrast to existing works on reconfigurable intelligent surface (RIS) aided ISAC systems using conventional phased arrays (PAs), this paper investigates a frequency diverse array (FDA)-enabled RIS-aided ISAC system, where the FDA aims to provide a distance-angle-dependent beampattern to effectively suppress the clutter, and RIS is employed to establish high-quality links between the BS and users/target. We aim to maximize sum rate by jointly optimizing the BS transmit beamforming vectors, the covariance matrix of the dedicated radar signal, the RIS phase shift matrix, the FDA frequency offsets and the radar receive equalizer, while guaranteeing the required signal-to-clutter-plus-noise ratio (SCNR) of the radar echo signal. To tackle this challenging problem, we first theoretically prove that the dedicated radar signal is unnecessary for enhancing target sensing performance, based on which the original problem is much simplified. Then, we turn our attention to the single-user single-target (SUST) scenario to demonstrate that the FDA-RIS-aided ISAC system always achieves a higher SCNR than its PA-RIS-aided counterpart. Moreover, it is revealed that the SCNR increment exhibits linear growth with the BS transmit power and the number of BS receive antennas. In order to effectively solve this simplified problem, we leverage the fractional programming (FP) theory and subsequently develop an efficient alternating optimization (AO) algorithm based on symmetric alternating direction method of multipliers (SADMM) and successive convex approximation (SCA) techniques. Numerical results demonstrate the superior performance of our proposed algorithm in terms of sum rate and radar SCNR.
\end{abstract}


\section{Introduction}

As one of the pioneering technologies for the upcoming sixth generation (6G) wireless systems, integrated sensing and communication (ISAC) has aroused extensive research interests in the recent past\cite{Intro_ref2}. ISAC technology enables wireless devices to perform high-speed data transmission and high-accuracy real-time sensing simultaneously through sharing the same resources, such as hardware platforms, spectrum, waveform and so on, thereby fulfilling stringent quality of service (QoS) requirements of numerous emerging 6G intelligent applications. In terms of the implementation of ISAC, the massive MIMO array with the capability of providing large spatial multiplexing/diversity gains and high-gain ultra-narrow beams has also been widely utilized to enhance both communication performance and spatial resolution for target detection\cite{Intro_ref2}. Unfortunately, considering that severe ISAC performance degradation may still be inevitable in harsh electromagnetic environments since the massive MIMO array can only passively adapt to the varying wireless environment, reconfigurable intelligent surface (RIS) has been proposed as a promising and cost-effective solution to actively reshape the wireless environment\cite{rismec}. Specifically, RIS consists of numerous low-cost reflecting elements, each of which can dynamically control the reflection amplitude and phase of the incident signal independently, thereby creating favorable propagation conditions and thus significantly improving ISAC performance\cite{Intro_ref4}.

The research works on RIS-aided ISAC systems have initially concentrated on a simple single-user multiple-input single-output (SU-MISO) scenario with a single target, where 
the signal-to-noise ratio (SNR) of the radar echo signal was optimized to maximize the target detection probability, while satisfying the minimum required SNR of the communication signal for the sake of guaranteeing  communication performance\cite{SUST1,SUST2}. However, this kind of scheme is inefficient in terms of wireless resources utilization including hardware platforms, spectrum and waveform. Moreover, although these works can provide fundamental insights on the effectiveness of RIS in enhancing ISAC performance, their design cannot be directly applied to general multi-user (MU) scenario in the presence of interference. Therefore, there have also been many studies focusing on exploring RIS-aided MU ISAC systems from both  communication-centric and radar-centric perspectives. For the   communication-centric design, optimization objectives such as system sum rate maximization, data mean square error (MSE) minimization and energy efficiency maximization have been widely investigated under various radar sensing constraints\cite{MUST1,MUST2,Intro_ref4,RIS_ISAC3}. Meanwhile, for the radar-centric design, the Cramer-Rao bound (CRB) and radar beampattern gain were often adopted as sensing performance evaluation metrics, while simultaneously regarding the communication QoS requirements as constraints\cite{MUST4,MUST_gain}. Considering that the target may be surrounded by clutters in realistic scenarios, the  signal-to-clutter-plus-noise ratio (SCNR) maximization was further studied in \cite{MUST5}, where a minimum level of signal-to-interference-plus-noise (SINR) for each user was also required. In order to further facilitate the deep integration of communication and sensing, the authors in \cite{ISAC_ref4} proposed  the weighted sum rate metric of both communication and sensing subsystems, and characterized the Pareto-optimal communication-sensing rate region to clearly illustrate the trade-off between these two functionalities. It is worth noting that the aforementioned studies all assume the conventional phased array (PA)-aided BS.  Since the conventional  PA can only generate the angle-dependent beampattern  with limited spatial degrees of freedom (DoFs),  the PA-aided ISAC system may  suffer from significant communication and sensing performance degradation, especially when the target and the clutters are in close spatial proximity.

Distinct from the PA of which an identical waveform is shared across all array elements, the frequency diverse array (FDA) assigns carrier frequency offset to each antenna to generate the distance-angle-dependent beampattern, thus attaining the additional range resolution as compared to the PA\cite{Intro_ref7}. Inspired by this appealing advantage, several works have investigated FDA-aided ISAC systems\cite{FDA_ISAC1,FDA_ISA2,FDA_ISAC3,FDA_ISAC4,FDA_ISAC5}. Initially, the spectrum coexistence of an FDA-based radar and a MIMO communication system was investigated in \cite{FDA_ISAC5}, where the FDA with uniform frequency offsets was only utilized for enhancing SCNR of the radar echo signal. For the dual-functional FDA-aided ISAC systems, the FDA can also be utilized for improving communication performance based on the additional DoFs induced by FDA frequency offsets\cite{FDA_ISAC1,FDA_ISAC4}. Specifically, in terms of the radar-centric design, the authors in \cite{FDA_ISAC1} utilized a set of orthogonal waveforms  to implement the primary radar detection function, and meanwhile the binary data bits were inserted into FDA frequency offsets for supporting the secondary communication function. From the communication-centric perspective, the authors in \cite{FDA_ISAC4} proposed a novel frequency offset permutation index modulation scheme to increase the communication data rate subject to the  required target estimation CRB. The studies in \cite{FDA_ISA2} further leveraged the FDA frequency offsets to create orthogonality between the radar and communication beams, and the communication signal was projected along the null space of the radar main beam to eliminate the cross-interference. The authors of \cite{FDA_ISAC3} considered a more general FDA-aided MU ISAC system, where the optimal beamforming design was investigated under both radar-centric and communication-centric schemes. It is worth noting that the above mentioned works are mostly limited to the single-user scenario\cite{FDA_ISAC1,FDA_ISA2,FDA_ISAC4}. Although the authors in \cite{FDA_ISAC3} have investigated the MU case, the optimization of frequency offsets is left unconsidered, and thus the potential of FDA in enhancing ISAC performance is far from fully exploited. Moreover, none of the above works considers introducing the RIS to assist communication and sensing, thus may suffer from severe performance deterioration in poor propagation conditions.

In order to effectively combat severe propagation loss and blockages, and meanwhile create more DoFs for both communication  and sensing functionalities, it is essential to integrate RIS and FDA into ISAC systems. To the best of our knowledge, there are few works investigating the joint deployment of FDA and RIS in ISAC systems.  An initial attempt for such FDA-RIS-aided wireless system implementation was presented in \cite{FDAandRIS}, where the FDA-based BS aimed at providing the distance-angle-dependent beampattern to prevent information leakage to the eavesdropper, and the RIS was deployed to establish high-quality links between the BS and legitimate users, thereby improving the system secrecy rate. Further than the work in \cite{FDAandRIS}, in this paper, we aim to maximize the system sum rate in the considered FDA-RIS-aided ISAC system by additionally optimizing the FDA frequency offsets, while guaranteeing the level of  SCNR of the radar echo signal. The main contributions are summarized as follows.
\begin{itemize}
    \item Firstly, we establish both communication and sensing channel models for the studied FDA-enabled RIS-aided (hereafter abbreviated as FDA-RIS-aided) ISAC system, which are usually more complicated than those of its PA-aided counterpart due to multiple transmit frequencies adopted by the FDA. Specifically, considering that the intrinsic time-variant beampattern of the FDA may cause severe communication and target detection performance degradation, we propose to apply the receive processing chain to effectively eliminate the time-variant components in propagation models. Accordingly, the equivalent time-invariant channel models can be established, based on which the corresponding optimization problem is formulated.
    \item   Secondly, to achieve an optimal balance between sensing and communication, we aim to maximize system sum rate while guaranteeing the required SCNR of the echo signal by jointly optimizing the BS transmit beamforming vectors, the covariance matrix of the dedicated radar signal, the RIS phase shift matrix, the FDA frequency offsets and the radar receive equalizer. We first theoretically prove that the dedicated radar signal is unnecessary for enhancing target sensing performance, based on which the original problem is much simplified. Then, we explore a simple single-user single-target (SUST) scenario to reveal that the FDA-aided BS always achieves a higher SCNR than its PA-aided counterpart, and demonstrate that the SCNR increment exhibits linear growth with the BS transmit power and the number of BS receive antennas.
    \item Finally, in order to effectively tackle this simplified problem, we firstly leverage the fractional programming (FP)-based technique to reformulate it into a tractable parametric subtraction form. We then propose an efficient alternating optimization (AO) algorithm jointly based on the symmetric alternating direction method of multiplier (SADMM) and successive convex approximation (SCA) methods. Under this AO framework, the optimal BS transmit beamforming vectors and receive equalizer can both be obtained in closed forms. Then, we apply the SADMM and SCA methods to find locally optimal solutions of the unit-modulus constrained RIS phase shifts and the FDA frequency offsets, respectively. Numerical results demonstrate the superior communication and sensing performance of our proposed algorithm.
\end{itemize}

\noindent {\textbf{Notations}}: Matrices and column vectors are written in boldface letters in upper-case and lower-case, respectively. ${\bf {I}}_{N}$ denotes the $N\!\times\!N$ identity matrix. For a complex number $x$, $\vert x\vert,x^*,\angle x$ denote its absolute value, conjugate and angle, respectively. For vector or matrix operation, $(\cdot)^{\rm T},(\cdot)^{\rm H}$ denote the transpose and conjugate transpose operations, $\Vert\bf x\Vert$ and $[{\bf x}]_m$ denote the $l_2$ norm and the $m$-th entry of vector $\bf x$. $\bf X^{\dagger}$, $\Vert{\bf X}\Vert$ and $[{\bf X}]_{m,n}$ denote the pseudo-inverse, the Frobenius norm and the $(m,n)$-th entry of matrix ${\bf X}$, respectively. The diagonal matrix composed of the elements of ${\bf x}$ and the vector composed of the diagonal elements of ${\bf X}$ are denoted by ${\rm diag}({\bf x})$ and ${\rm diag}({\bf X})$, respectively. ${\rm vec}({\bf X})$ represents stacking the elements of ${\bf X}$ into a vector. The symbol $\otimes$  and $\circ$ denote the Kronecker and Hadamard product operation, respectively. The abbreviation ``w.r.t.'' refers to the phrase ``with respect to''. $\mathcal{O}(\cdot)$ denotes the standard big-O notation.

\section{System Model and Problem Formulation}\label{S_sysandpro}

{\noindent\textit{A. FDA-RIS-aided  ISAC System Setup}}

As shown in Fig.~\ref{System_Model}, in this paper, we consider an FDA-RIS-aided ISAC system, where the BS  equipped with an FDA-aided ISAC transmitter and  a ULA-based radar receiver  intends to communicate with $K$ single-antenna users and to detect a target of interest surrounded by $C$ clutters
simultaneously with  the assistance of a RIS. The FDA-aided ISAC  transmitter and  the radar receiver are  equipped with $N_{\rm t}$ transmit  and   $N_{\rm r}$  receive antennas (denoted by the sets $\mathcal{N}_{\rm t}\!=\!\{1,\cdots,N_{\rm t}\}$ and $\mathcal{N}_{\rm r}\!=\!\{1,\cdots,N_{\rm r}\}$), respectively, and operate on the same hardware platform.  Different from  the traditional PA with only angle-dependent beampattern due to its identical waveform on each transmit antenna, the FDA deployed at the BS is able to provide the distance-angle-dependent beampattern by introducing a frequency offset on each transmit antenna. To be specific, the radiated signal frequency at the $n_{\rm t}$-th transmit antenna  is expressed  as\cite{Intro_ref7}
\begin{align}
    f_{ n_{\rm t}} = f_{\rm ref}+\Delta f_{ n_{\rm t}},~~\forall n_{\rm t}\in\mathcal{N}_{\rm t},
\end{align}
where $f_{\rm ref}$  and $\Delta f_{ n_{\rm t}}$ denote the reference carrier frequency and the frequency offset at the $n_{\rm t}$-th transmit antenna, respectively. Furthermore, the RIS consists of $M\!=\!M_{\rm azi}\!\times\!M_{\rm ele}$ reflecting elements (denoted by the set $\mathcal{M}\!\triangleq\!\{1,\cdots,M\}$), where $M_{\rm azi}$ and $M_{\rm ele}$ represent the numbers of azimuth and elevation elements, respectively. Since each RIS element is able to reflect the incident signal using a controllable phase shift, we define ${\bm\Theta}\!=\!{\rm diag}\{e^{j\theta_1},\cdots,e^{j\theta_{M}}\}\!
\in\!\mathbb{C}^{M\times M}$ as the diagonal  reflection coefficient matrix at the RIS, where $\theta_{m}\!\in\![0,2\pi]$ denotes the phase shift imposed on the incident signal at the $m$-th RIS element. It is  assumed  that  direct  links between the BS and all users/target/clutters are blocked by obstacles such as buildings and trees. 
 To circumvent this blockage, we consider deploying the RIS in an open area to create LoS channels between the RIS and the BS/target, which promotes the successful sensing\cite{BR_LOS1}. 

 \begin{figure}[t]
    \vspace{-16mm}
    \begin{minipage}[t]{0.5\linewidth}
        \centering
        \includegraphics[width=0.95\textwidth]{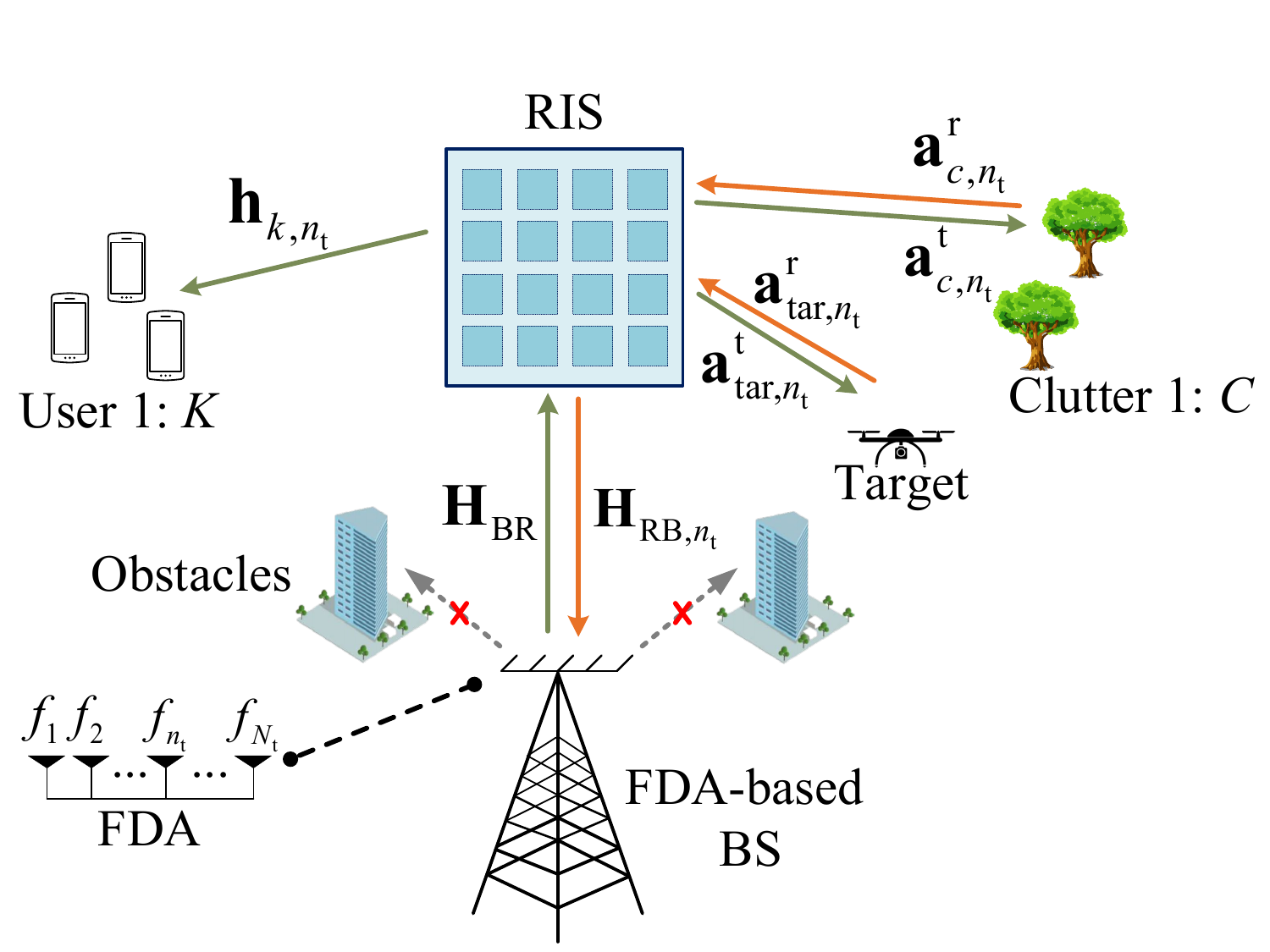}
        \vspace{-7mm}
        \caption{The considered FDA-RIS-aided ISAC system.}
        \label{System_Model}
    \end{minipage}%
    \begin{minipage}[t]{0.5\linewidth}\vspace{-52mm}
        \centering
        \includegraphics[width=1\textwidth]{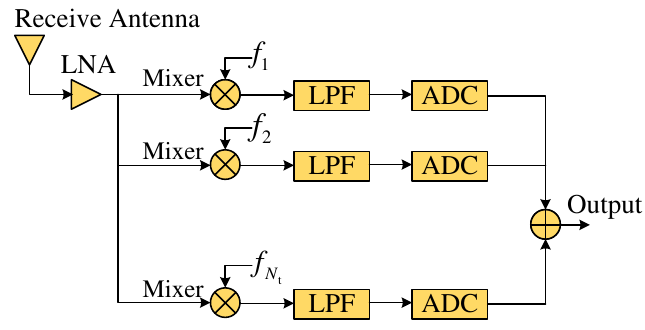}
        \vspace{-4.8mm}
        \caption{Receive processing chain for the FDA-RIS-aided ISAC system.}
        \label{Fig_LPF}
    \end{minipage}
    \vspace{-8mm}
\end{figure}

  For  the  sake  of performing communication and sensing simultaneously,  the FDA-aided BS intends to  send  $K$ desired information symbols $s_k,k\!\in\!\mathcal{K}\!\triangleq\!\{1,\cdots,K\}$ to CUs, together
 with a dedicated  radar signal ${\bf s}_{0}\!\in\!\mathbb{C}^{N_{\rm t}\times 1}$, which are all assumed to be complex Gaussian distributed, i.e., $s_k\!\sim\!\mathcal{CN}(0,1)$ and ${\bf s}_{0}\!\sim\!\mathcal{CN}({\bf 0},{\bf R}_{0})$ with ${\bf R}_{0}\!\in\!\mathbb{C}^{N_{\rm t}\times N_{\rm t}}$ being the positive semidefinite covariance  matrix,  and independent of  each other.  Let ${\bf w}_k\!\in\!\mathbb{C}^{N_{\rm t}\times 1}$ denote the transmit beamforming vector for the CU $k$,  the 
transmitted signal by the FDA-aided BS is then expressed as 
\begin{align}
    {\bf x}\!=\!\sum\limits_{k=1}^{K}{\bf w}_ks_k\!+\!{\bf s}_{0},
\end{align}
which is subject to a total transmit power constraint, i.e., $\mathbb{E}[\Vert{\bf x}\Vert^2]\!=\!\sum\limits_{k=1}^{K}\Vert{\bf w}_k\Vert^2\!+\!{\rm Tr}({\bf R}_0)\!\leq\!P_{\rm B}$ , with $P_{\rm B}$ being the maximum transmit power  budget at the BS. Nevertheless, it is known that  the transmit beampattern of FDA is  actually  time-variant, which may cause inevitable  communication and target detection   performance degradation \cite{FDA_LPF1}. Fortunately, the receive processing chain has been proposed as an effective remedy to alleviate this issue \cite{FDA_LPF1,FDA_LPF4}. As shown in Fig.~\ref{Fig_LPF},  the received signal is firstly amplified by a low-noise amplifier (LNA), and then separated into $N_{\rm t}$ streams. The $n_{\rm t}$-th stream is firstly mixed with a reference signal at a frequency $f_{n_{\rm t}}$, which is then filtered by a low-pass filter (LPF) to eliminate its time-variant component, and subsequently sampled by an analog-to-digital converter (ADC). Finally, the $N_{\rm t}$ sampled streams are summed together to generate the output.
% %%%参考这篇文章的描述，简要介绍工作原理：
% As shown in Fig. 2, the received signals are first amplified in a low-noise amplifier. Then, the signals are separated
% by a set of bandpass filters (BPFs) with bandwidth f and center frequencies fm, m = 1, 2, ... , M. After this implementation, the signals transmitted by each transmit element
% are independently obtained.

{\noindent\textit{B. Communication Signal Model}}

In  this FDA-RIS-aided ISAC system, we assume narrowband  transmission over quasi-static flat fading channels.  Let  ${\bf H}_{\rm BR}\!\in\!\mathbb{C}^{M\times N_{\rm t}}$  denote the channel between the BS and the RIS.  Inspired by  the above-mentioned proper deployment of the RIS and the employment  of the receive processing chain,  we can model ${\bf H}_{\rm BR}$  as  the  following LoS  channel
\begin{align}\label{Chan-BR}
   %%%% 确认是否需要azimuth 和 elevation的天线索引
[{\bf H}_{\rm BR}]_{m,n_{\rm t}}\!=\!\sqrt{\beta_{\rm BR}}e^{-j{2\pi}\phi_{m,n_{\rm t}}^{\rm BR}},m\!\in\!\mathcal{M},n_{\rm t}\!\in\!\mathcal{N_{\rm t}},
\end{align}
where $\beta_{\rm BR}$ denotes the large-scale path loss of the BS-RIS link. $\phi_{m,n_{\rm t}}^{\rm BR}$ represents the  phase shift from the $n_{\rm t}$-th BS  antenna to the $m$-th RIS  element, which is computed as
\begin{align}\label{Chan-BR-app}
    \phi_{m,n_{\rm t}}^{\rm BR}\!&=\!\frac{f_{n_{\rm t}}}{\rm c}(D_{\rm BR}\!-\!(n_{\rm t}\!-\!1)d_{\rm B}\sin\psi^{\rm BR}\!+\!(m_{\rm ele}\!-\!1)d_{\rm ele}\sin\theta_{\rm azi}^{\rm BR}\sin\theta_{\rm ele}^{\rm BR}\notag\\
    &+\!(m_{\rm azi}\!-\!1)d_{\rm azi}\sin\theta_{\rm azi}^{\rm BR}\cos\theta_{\rm ele}^{\rm BR})\notag\\
    &\overset{(a)}{\approx}\!\frac{f_{n_{\rm t}}}{\rm c}D_{\rm BR}\!-\!\frac{f_{\rm ref}}{\rm c}((n_{\rm t}\!-\!1)d_{\rm B}\sin\psi^{\rm BR}\!-\!(m_{\rm ele}\!-\!1)d_{\rm ele}\sin\theta_{\rm azi}^{\rm BR}\sin\theta_{\rm ele}^{\rm BR}\!\notag\\
    &-\!(m_{\rm azi}\!-\!1)d_{\rm azi}\sin\theta_{\rm azi}^{\rm BR}\cos\theta_{\rm ele}^{\rm BR}),
\end{align}
where $m_{\rm ele}\!=\!{\rm mod}(m,M_{\rm ele})$ and $m_{\rm azi}\!=\!\lfloor m/M_{\rm ele}\rfloor$ denote the indices of the $m$-th RIS reflecting element along the elevation and azimuth directions, respectively. $D_{\rm BR}$ denotes the propagation distance of the BS-RIS link and ${\rm c}$ is the speed of light. $d_{\rm B}\!=\!\frac{\lambda_{\rm wav}}{2}\!=\!\frac{\rm c}{2f_{\rm ref}}$ denotes the half-wavelength inter-antenna spacing at the BS. $d_{\rm azi}(d_{\rm ele})\!=\!\frac{\lambda_{\rm wav}}{8}$  represents  the sub-wavelength inter-element spacing at the RIS along the azimuth (elevation) direction. $\psi^{\rm BR}$, $\theta_{\rm azi}^{\rm BR}$ and $\theta_{\rm ele}^{\rm BR}$ represent the angle-of-departure (AoD), the azimuth and elevation angles-of-arrival (AoAs) of the BS-RIS link, respectively. In particular, in \eqref{Chan-BR-app}, the approximation $(a)$ holds due to $\Delta f_{n_{\rm t}}\!\ll\!f_{\rm ref}$\cite{FDA_LPF1}.

Furthermore, owing to the  $N_{\rm t}$  transmit frequencies utilized at the BS and the receive processing chain adopted at the users,  we  denote ${\bf h}_{k,n_{\rm t}}\!\in\!\mathbb{C}^{M\times 1}$ as the channel between the RIS and user $k$ corresponding to the FDA frequency $f_{n _{\rm t}}$. Without loss of generality, the channel ${\bf h}_{k,n_{\rm t}}$ is assumed to be Rician fading, thus yielding 
\begin{align}
    {\bf h}_{k,n_{\rm t}}\!=\!\sqrt{\frac{\kappa\beta_{{\rm RU},k}}{\kappa+1}}{\bf h}_{k,n_{\rm t}}^{\rm LoS}\!+\!\sqrt{\frac{\beta_{{\rm RU},k}}{\kappa+1}}{\bf h}_{k,n_{\rm t}}^{\rm NLoS},
\end{align}
where $\kappa$ and $\beta_{{\rm RU},k}$ denote the Rician factor and the large-scale path loss  of the link between the RIS and user $k$, respectively. ${\bf h}_{k,n_{\rm t}}^{\rm LoS}$ and ${\bf h}_{k,n_{\rm t}}^{\rm NLoS}\!\sim\!\mathcal{CN}({\bf 0},{\bf I}_{M})$ represent the LoS and NLoS components, respectively. Similar to \eqref{Chan-BR}, the LoS component ${\bf h}_{k,n_{\rm t}}^{\rm LoS}$ is written as
\begin{align}\label{Chan-RU}
    \big[{\bf h}_{k,n_{\rm t}}^{\rm LoS}]_{m}\!=\!e^{-j{2\pi}\phi_{k,n_{\rm t},m}^{{\rm RU}}}, m\!\in\!\mathcal{M},
\end{align}
where the phase shift $\phi_{k,n_{\rm t},m}^{{\rm RU}}$ is defined similarly to \eqref{Chan-BR-app} and thus given by
\begin{align}
    \phi_{k,n_{\rm t},m}^{{\rm RU}}\!=\!-\frac{f_{n_{\rm t}}}{\rm c}D_{{\rm RU},k}\!+\!\frac{f_{\rm ref}}{\rm c}((m_{\rm ele}\!-\!1)d_{\rm ele}\sin\theta_{\rm azi}^{{\rm RU},k}\sin\theta_{\rm ele}^{{\rm RU},k}\!+\!(m_{\rm azi}\!-\!1)d_{\rm azi}\sin\theta_{\rm azi}^{{{\rm RU},k}}\cos\theta_{\rm ele}^{{\rm RU},k}),
\end{align}
where $D_{{\rm RU},k}$ denotes the propagation distance between the RIS and user $k$. $\theta^{{\rm RU},k}_{{\rm azi}}$ and $\theta^{{\rm RU},k}_{{\rm ele}}$ represent the azimuth and elevation AoDs of the link between the RIS and user $k$, respectively.

Based on \eqref{Chan-BR} and \eqref{Chan-RU}, the received signal at the user $k$ can be expressed as 
\begin{align}\label{comm_sig}
    y_k\!=\!\sum\limits_{n_{\rm t}=1}^{N_{\rm t}}\!{\bf h}_{k,n_{\rm t}}^{\rm H}{\bm\Theta}{\bf h}_{{\rm BR},n_{\rm t}}[{\bf x}]_{n_{\rm t}}\!+\!n_k\!=\!\tilde{\bf h}_k^{\rm H}{\bf x}\!+\!n_{k},
\end{align}
where ${\bf h}_{{\rm BR},n_{\rm t}}$ denotes the $n_{\rm t}$-th column of ${\bf H}_{\rm BR}$ corresponding to the $n_{\rm t}$-th FDA transmit antenna, $\tilde{\bf h}_k\!=\![{\bf h}_{k,1}^{\rm H}{\bm\Theta}{\bf h}_{{\rm BR},1},\cdots,{\bf h}_{k,N_{\rm t}}^{\rm H}{\bm\Theta}{\bf h}_{{\rm BR},N_{\rm t}}]^{\rm H}\!\in\!\mathbb{C}^{N_{\rm t}\times 1}$ denotes the  BS-RIS-User cascaded channel for the user $k$, and $n_k\!\sim\!\mathcal{CN}(0,\sigma_k^2)$ denotes the additive Gaussian noise at the user $k$.

{\noindent\textit{C. Radar Signal Model}}

As for the radar subsystem, we assume that the BS is able to detect the 
 target of interest located at $(\theta^{{\rm RT}}_{\rm azi},\theta^{{\rm RT}}_{\rm ele})$ within a distance of $D_{\rm RT}$ from the RIS in the presence of $C$ clutters, denoted by the set $\mathcal{C}\!\triangleq\!\{1,\cdots,C\}$, where the clutter $c$ is located at $(\theta^{{\rm RC},c}_{\rm azi},\theta^{{\rm RC},c}_{\rm ele})$ within a distance $D_{{\rm RC},c}$ from the RIS.  In  contrast  to the communication subsystem, where the existence of clutters  is beneficial  to  enhance spatial diversity, in this radar subsystem, the echo signals from the clutters are detrimental to target sensing and need to be mitigated. Owing to $N_{\rm t}$ transmit frequencies utilized at the BS, we define ${\bf a}_{{\rm tar},n_{\rm t}}^{\rm t}\!\in\!\mathbb{C}^{M\times 1}$ and ${\bf a}_{{\rm tar},n_{\rm t}}^{\rm r}\!\in\!\mathbb{C}^{M\times 1}$ as the transmit steering vector of the RIS-target link and the receive steering vector of the target-RIS link corresponding to the FDA frequency $f_{n_{\rm t}}$, respectively, which are modelled as
 \begin{align}\label{chan_RT}
    [{\bf a}_{{\rm tar},n_{\rm t}}^{\rm t}]_m\!=\!e^{-j2\pi\phi_{{\rm tar},n_{\rm t},m}^{\rm t}},~[{\bf a}_{{\rm tar},n_{\rm t}}^{\rm r}]_m\!=\!e^{-j2\pi\phi_{{\rm tar},n_{\rm t},m}^{\rm r}},~m\!\in\!\mathcal{M},
 \end{align}
 where the phase shifts $\phi_{{\rm tar},n_{\rm t},m}^{\rm t}$ and $\phi_{{\rm tar},n_{\rm t},m}^{\rm r}$ are given by
 \begin{align}
    \phi_{{\rm tar},n_{\rm t},m}^{\rm t}\!&=\!-\frac{f_{n_{\rm t}}}{\rm c}D_{{\rm RT}}\!+\!\frac{f_{\rm ref}}{\rm c}((m_{\rm ele}\!-\!1)d_{\rm ele}\sin\theta_{\rm azi}^{{\rm RT}}\sin\theta_{\rm ele}^{{\rm RT}}\!+\!(m_{\rm azi}\!-\!1)d_{\rm azi}\sin\theta_{\rm azi}^{{{\rm RT}}}\cos\theta_{\rm ele}^{{\rm RT}}),\\
    \phi_{{\rm tar},n_{\rm t},m}^{\rm r}\!&=\!\frac{f_{n_{\rm t}}}{\rm c}D_{{\rm RT}}\!+\!\frac{f_{\rm ref}}{\rm c}((m_{\rm ele}\!-\!1)d_{\rm ele}\sin\theta_{\rm azi}^{{\rm RT}}\sin\theta_{\rm ele}^{{\rm RT}}\!+\!(m_{\rm azi}\!-\!1)d_{\rm azi}\sin\theta_{\rm azi}^{{{\rm RT}}}\cos\theta_{\rm ele}^{{\rm RT}}).
 \end{align}
 In addition, ${\bf a}_{{c},n_{\rm t}}^{\rm t}$ $({\bf a}_{{c},n_{\rm t}}^{\rm r})$ denotes the transmit (receive) steering vector corresponding to the clutter $c$ and is defined similarly to ${\bf a}_{{\rm tar},n_{\rm t}}^{\rm t}$ (${\bf a}_{{\rm tar},n_{\rm t}}^{\rm r}$).

Since the FDA-based BS adopts $N_{\rm t}$ transmit frequencies and the radar receiver  employs a ULA  with the receive processing  chain shown in Fig.~\ref{Fig_LPF}, the effective RIS-BS  channel associated with the FDA frequency $f_{n_{\rm t}}$ can be expressed as
\begin{align}\label{chan-RB}
    [{\bf H}_{{\rm RB},n_{\rm t}}]_{n_{\rm r},m}\!=\!\sqrt{\beta_{\rm BR}}e^{-j{2\pi}\phi_{n_{\rm t},n_{\rm r},m}^{\rm RB}},m\!\in\!\mathcal{M},n_{\rm r}\!\in\!\mathcal{N_{\rm r}},
\end{align}
where the phase shift $\phi_{n_{\rm t},n_{\rm r},m}^{\rm RB}$ is given by
\begin{align}
    \phi_{n_{\rm t},n_{\rm r},m}^{\rm RB}\!&=\!\frac{f_{n_{\rm t}}}{\rm c}D_{\rm BR}\!+\!\frac{f_{\rm ref}}{\rm c}((n_{\rm r}\!-\!1)d_{\rm B}\sin\psi^{\rm BR}\!-\!(m_{\rm ele}\!-\!1)d_{\rm ele}\sin\theta_{\rm azi}^{\rm BR}\sin\theta_{\rm ele}^{\rm BR}\!\notag\\
    &-\!(m_{\rm azi}\!-\!1)d_{\rm azi}\sin\theta_{\rm azi}^{\rm BR}\cos\theta_{\rm ele}^{\rm BR}).
\end{align}

In this  radar subsystem,  both  communication   and  radar  signals  can be  utilized for detecting the  target of  interest,  both of which undergo two types of  RIS-reflected echo links, i.e., the BS-RIS-target-RIS-BS cascaded echo link  and the BS-RIS-clutter $c$-RIS-BS cascaded echo link.  Consequently, the received  echo signal at the BS is written as
\begin{align}\label{radar_sig}
    {\bf y}_{\rm R}\!&=\!\xi_{\rm T}\sum\limits_{n_{\rm t}=1}^{N_{\rm t}}\!{\bf H}_{{\rm RB},n_{\rm t}}{\bm \Theta}{\bf a}_{{\rm tar},n_{\rm t}}^{\rm r}({\bf a}_{{\rm tar},n_{\rm t}}^{\rm t})^{\rm H}{\bm\Theta}{\bf h}_{{\rm BR},n_{\rm t}}[{\bf x}]_{n_{\rm t}}\notag\\
    &+\!\sum\limits_{c=1}^{ C}\sum\limits_{n_{\rm t}=1}^{N_{\rm t}}\!\xi_c{\bf H}_{{\rm RB},n_{\rm t}}{\bm \Theta}{\bf a}_{c,n_{\rm t}}^{\rm r}({\bf a}_{c,n_{\rm t}}^{\rm t})^{\rm H}{\bm\Theta}{\bf h}_{{\rm BR},n_{\rm t}}[{\bf x}]_{n_{\rm t}}\!+{\bf n}_{\rm R},\notag\\
    &=\!\xi_{\rm T}{\bf H}_{\rm BT}{\bf x}\!+\!\sum\limits_{c=1}^{ C}\xi_c{\bf H}_{{\rm BC},c}{\bf x}\!+\!{\bf n}_{\rm R}
\end{align}
where ${\bf H}_{\rm BT}\!=\![{\bf H}_{{\rm RB},1}{\bm \Theta}{\bf a}_{{\rm tar},1}^{\rm r}({\bf a}_{{\rm tar},1}^{\rm t})^{\rm H}{\bm\Theta}{\bf h}_{{\rm BR},1},\cdots,{\bf H}_{{\rm RB},N_{\rm t}}{\bm \Theta}{\bf a}_{{\rm tar},N_{\rm t}}^{\rm r}({\bf a}_{{\rm tar},N_{\rm t}}^{\rm t})^{\rm H}{\bm\Theta}{\bf h}_{{\rm BR},N_{\rm t}}]\!\in\!\mathbb{C}^{N_{\rm r}\times N_{\rm t}}$ and ${\bf H}_{{\rm BC},c}\!=\![{\bf H}_{{\rm RB},1}{\bm \Theta}{\bf a}_{{c},1}^{\rm r}({\bf a}_{{c},1}^{\rm t})^{\rm H}{\bm\Theta}{\bf h}_{{\rm BR},1},\cdots,{\bf H}_{{\rm RB},N_{\rm t}}{\bm \Theta}{\bf a}_{{c},N_{\rm t}}^{\rm r}({\bf a}_{{c},N_{\rm t}}^{\rm t})^{\rm H}{\bm\Theta}{\bf h}_{{\rm BR},N_{\rm t}}]\!\in\!\mathbb{C}^{N_{\rm r}\times N_{\rm t}}$ represent the BS-RIS-target-RIS-BS and BS-RIS-clutter $c$-RIS-BS echo links, respectively. $\xi_{\rm T}\!\sim\!\mathcal{CN}(0,\beta_{\rm T})$ and $\xi_{c}\!\sim\!\mathcal{CN}(0,\beta_{c})$ denote the overall random complex gains associated with the target and clutter $c$, respectively, incorporating both the two-way channel large-scale path loss and the fluctuating radar cross section (RCS) coefficient. In addition, ${\bf n}_{\rm R}\!\sim\!\mathcal{CN}({\bf 0},\sigma_{\rm R}^2{\bf I}_{N_{\rm r}})$ denotes the additive Gaussian noise at the BS receive ULA.

{\noindent\textit{D. Performance Metrics and Problem Formulation}}

{\indent{\textit{1) Communication Performance Metric:}}} 
For the communication subsystem, we adopt the achievable
 sum rate to measure the  system  transmission performance.  To be specific, it follows from \eqref{comm_sig} that the data rate of user $k$ can be  written as 
\begin{align}
    R_k\!=\!\log_2\Big(1+\frac{\vert\tilde{\bf h}_k^{\rm H}{\bf w}_k\vert^2}{\sum\limits_{{k'=1,k'\neq k}}^{K}\vert\tilde{\bf h}_k^{\rm H}{\bf w}_{k'}\vert^2\!+\tilde{\bf h}_k^{\rm H}{\bf R}_0\tilde{\bf h}_k\!+\!\sigma_k^2}\Big).
\end{align}
Then, the achievable sum rate of the  considered  FDA-RIS-aided  ISAC system  is given by
\begin{align}\label{Rsum}
    R_{\rm sum}\!=\!\sum\limits_{k=1}^{K} R_k.
\end{align}

{\indent{\textit{2) Radar Performance Metric:}}} 
For the radar subsystem, we adopt the SCNR of the echo signal to assess the detectability of a target in the presence of clutters, considering  that the target detection probability  generally monotonically increases with SCNR.  %%% SCNR yu probability de guan xi.}
Based on \eqref{radar_sig}, the SCNR of the echo signal is expressed as\cite{SCNR_ref}
\begin{align}\label{SCNR_formula}
    {\Gamma}\!&=\!\frac{\mathbb{E}[\vert\xi_{\rm T}{\bf u}^{\rm H}{\bf H}_{\rm BT}{\bf x}\vert^2]}{\mathbb{E}[\,\vert\!\sum\nolimits_{c=1}^{C}\xi_{\rm c}{\bf u}^{\rm H}{\bf H}_{{\rm BC},c}{\bf x}\vert^2\!+\!\vert{\bf u}^{\rm H}{\bf n}_{\rm R}\vert^2]}\notag\\
    &=\!\frac{\beta_{\rm T}\big(\sum\limits_{k=1}^{K}\vert{\bf u}^{\rm H}{\bf H}_{\rm BT}{\bf w}_k\vert^2\!+\!{\bf u}^{\rm H}{\bf H}_{\rm BT}{\bf R}_0{\bf H}_{\rm BT}^{\rm H}{\bf u}\big)}{\sum\limits_{c=1}^{C}\!\beta_c\big(\sum\limits_{k=1}^{K}\vert{\bf u}^{\rm H}{\bf H}_{{\rm BC},c}{\bf w}_k\vert^2\!+\!{\bf u}^{\rm H}{\bf H}_{{\rm BC},c}{\bf R}_0{\bf H}_{{\rm BC},c}^{\rm H}{\bf u}\big)\!+\!\Vert{\bf u}\Vert^2\sigma_{\rm R}^2}, 
\end{align}
where ${\bf u}\!\in\!\mathbb{C}^{N_{\rm r}\times1}$ denotes the radar receive equalizer.

{\indent{\textit{3) Problem Formulation:}}} 
In  the  FDA-RIS-aided ISAC system, we aim to maximize the achievable sum rate in \eqref{Rsum} while guaranteeing the level of SCNR of the echo signal by jointly optimizing the BS transmit beamforming vectors ${\bf w}_k$, the covariance matrix of the dedicated radar signal ${\bf R}_0$, the RIS phase shift matrix ${\bm \Theta}$, the FDA frequency offsets $\Delta f_{n_{\rm t}}$ and the radar receive equalizer ${\bf u}$. Mathematically, the corresponding optimization problem is formulated as
\begin{subequations}\label{P1}
    \begin{align}
        ({\text{P1}}):~  \underset{{{\mathcal{V}_{\rm P_1}}}}
        {{\max}}~~&R_{\rm sum}\label{P1obj} \\
        {\rm {s.t.}}~~&\sum\limits_{k=1}^{K}\Vert {\bf w}_{ k}\Vert^2\!+\!{\rm Tr}({\bf R}_0)\leq P_{\rm B},\label{P1cons1} \\
        &{\Gamma}\geq \gamma_{\rm T},\label{P1cons2}\\
        &\vert[{\bm \Theta}]_{m,m}\vert=1,\forall m\in\mathcal{M}\!\triangleq\!\{1,\cdots,M\},\label{P1cons3}\\
        &\Delta f_{ n_{\rm t}}\in[0,f_{\max}],\forall n_{\rm t}\in\mathcal{N}_{\rm t},\label{P1cons4}
    \end{align}
\end{subequations}
where ${\mathcal{V}_{\rm P_1}}\!\!=\!\!\{{\bf w}_{ k},{\bf R}_0\!\succeq\!{\bf 0},{{\bm \Theta}},\Delta f_{ n_{\rm t}},\!{\bf u}\}$.  \eqref{P1cons1} denotes the BS transmit power constraint.  \eqref{P1cons2} guarantees the SCNR of the echo signal above the threshold $\gamma_{\rm T}$. \eqref{P1cons3} represents the unit-modulus constraint imposed on RIS phase shifts.  \eqref{P1cons4} models the allowable frequency offset interval at each FDA antenna. Unfortunately, problem (P1) is non-convex and challenging to solve due to the strongly coupled variables in $\mathcal{V}_{\rm P_1}$, the complicated log-fractional form of the objective function \eqref{P1obj} and the non-convex constraints \eqref{P1cons2} and \eqref{P1cons3}. To make this problem tractable, we first simplify the problem (P1) in the following section.

\iffalse
\begin{figure}[t]
    \begin{minipage}[t]{0.5\linewidth}
        \centering
        \includegraphics[width=1\textwidth]{Fig/system_model.pdf}
        \caption{The considered FDA-enabled RIS-assisted \\ISAC system.}
        \label{System_Model}
    \end{minipage}%
    \begin{minipage}[t]{0.5\linewidth}
        \centering
        \includegraphics[width=0.85\textwidth]{Fig/LPF.pdf}
        \caption{Processing chain for the received signal for the FDA-enabled ISAC system.}
        \label{Fig_LPF}
    \end{minipage}
\end{figure}
\fi

\section{Equivalent Simplification and Performance Discussion}\label{FDAPA}

In this section, by exploring the characteristics of the sum rate maximization problem (P1), we first theoretically prove that the dedicated radar signal is unnecessary for  enhancing  target sensing performance. Motivated by this fact, the original optimization problem (P1) is substantially simplified. To obtain insights into the considered FDA-RIS-aided ISAC system, we then investigate a simplified single-user single-target (SUST) scenario, where we analytically reveal the advantages of FDA in enhancing SCNR as compared to conventional PA.

{\noindent\textit{A. Equivalent Simplification}}

It follows from problem (P1) that the transmit beamforming vectors ${\bf w}_k,\forall k\!\in\!\mathcal{K}$ and the covariance matrix of the dedicated radar signal ${\bf R}_0$ are coupled in both the objective function and constraints, which makes it challenging to address. Fortunately, by analyzing its inherent characteristics, we have the following proposition.
\begin{proposition}\label{prop-sim}
    There must exist an optimal solution ${\mathcal{V}_{\rm P_1}^{\rm opt}}\!\!=\!\!\{{\bf w}_{ k}^{\rm opt}, {\bf R}_{0}^{\rm opt}, {{\bm \Theta}}^{\rm opt},\Delta f_{ n_{\rm t}}^{\rm opt},{\bf u}^{\rm opt}\}$ with ${\bf R}_{0}^{\rm opt}\!=\!{\bf 0}$ to problem (P1).
\end{proposition}
\begin{proof}
    Please refer to Appendix~\ref{app-prop-sim}.
\end{proof}

Based on Proposition~\ref{prop-sim}, problem (P1) can be equivalently simplified as 
\begin{subequations}\label{P1-1}
    \begin{align}
        ({\text{P1-1}}):~  \underset{{{\mathcal{V}_{\rm P_{1-1}}}}}
        {{\max}}~&\hat{R}_{\rm sum}\!=\!\sum\nolimits_{k=1}^{K}\!\log_2\Big(1\!+\!\frac{\vert\tilde{\bf h}_k^{\rm H}{\bf w}_k\vert^2}{\sum\nolimits_{k'\neq k}\vert\tilde{\bf h}_k^{\rm H}{\bf w}_{k'}\vert^2\!+\!\sigma_k^2}\Big)\label{P1-1obj} \\
        {\rm {s.t.}}~~&\sum\nolimits_{k=1}^{K}\Vert {\bf w}_{ k}\Vert^2\leq P_{\rm B},\label{P1-1cons1} \\
        &\frac{\beta_{\rm T}\sum\nolimits_{k=1}^{K}\vert{\bf u}^{\rm H}{\bf H}_{\rm BT}{\bf w}_k\vert^2\!}{\sum\nolimits_{c=1}^{C}\!\sum\nolimits_{k=1}^{K}\beta_c\vert{\bf u}^{\rm H}{\bf H}_{{\rm BC},c}{\bf w}_k\vert^2\!+\!\Vert{\bf u}\Vert^2\sigma_{\rm R}^2}\geq \gamma_{\rm T},\label{P1-1cons2}\\
        &\eqref{P1cons3},\eqref{P1cons4}\label{P1-1cons3},
    \end{align}
\end{subequations}
where ${\mathcal{V}_{\rm P_{1-1}}}\!\!=\!\!\{{\bf w}_{ k},{{\bm \Theta}},\Delta f_{ n_{\rm t}},\!{\bf u}\}$. Unfortunately, problem (P1-1) is still hard to solve due to its non-convexity arising from the objective function and constraints \eqref{P1cons3} and \eqref{P1-1cons2}, as well as the coupling among multiple optimization variables. As such, it is difficult to obtain any useful insights into the considered FDA-RIS-aided ISAC system from solving problem (P1-1) directly.  To overcome this issue,  in the following subsection, we  turn our attention to a simple SUST scenario to theoretically demonstrate superior performance of the  FDA-aided ISAC system over its conventional  PA-aided counterpart.

{\noindent\textit{B. Performance Discussion}}

In this subsection, we focus on the SUST FDA-RIS-aided ISAC system, where the FDA-aided BS communicates with a single user in the presence of a single clutter with the aid of the RIS, implying that the user is also the target to be sensed (hereafter referred to as ``target" to avoid confusion). In particular, we assume that the target and the clutter are located at the same direction, i.e., $\theta^{{\rm RT}}_{\rm azi}\!=\!\theta^{{\rm RC}}_{\rm azi}$ and $\theta^{{\rm RT}}_{\rm ele}\!=\!\theta^{{\rm RC}}_{\rm ele}$, for the sake of clearly  demonstrating the 
superior clutter suppression performance of  the FDA  with distance-angle-dependent beampattern  over  the  conventional  PA. Moreover, the uniform transmit frequency offset is adopted for the FDA antennas for simplicity, i.e., $f_{n_{\rm t}}\!=\!f_{\rm fre}\!+\!(n_{\rm t}\!-\!1)\Delta f$ with $\Delta f$ representing the frequency offset increment\cite{Samefre}. The BS transmit beamforming vectors ${\bf w}_k$, the channel ${\bf H}_{{\rm BC},c}$,  the angles  $\{{\theta_{\rm azi}^{{\rm RC},c}},{\theta_{\rm ele}^{{\rm RC},c}}\}$ and the distance  $D_{{\rm RC},c}$  reduce to ${\bf w}_{\rm c}$, ${\bf H}_{{\rm BC}}$, $\{{\theta_{\rm azi}^{{\rm RC}}},{\theta_{\rm ele}^{{\rm RC}}}\}$ and $D_{{\rm RC}}$, respectively. Accordingly,   the   SCNR for this SUST system is given by 
\begin{align}\label{SUST_gam}
    \Gamma_{\rm SUST}\!&=\!\frac{\beta_{\rm T}\vert{\bf u}^{\rm H}{\bf H}_{\rm BT}{\bf w}_{\rm c}\vert^2}{\beta_{\rm C}\vert{\bf u}^{\rm H}{\bf H}_{{\rm BC}}{\bf w}_{\rm c}\vert^2\!+\!\Vert{\bf u}\Vert^2\sigma_{\rm R}^2}\notag\\
    &=\!\frac{{\beta}_{\rm BR}^2{\beta}_{\rm T}\vert{\bf u}^{\rm H}{\bf b}_{\rm RB,r}\vert^2\vert{\bf b}_{\rm BR}^{\rm H}{\bm\Theta}{\bf a}_{\rm tar,0}^{\rm r}({\bf a}_{\rm tar,0}^{\rm t})^{\rm H}{{\bm\Theta}}{\bf b}_{\rm BR}\vert^2\vert{\bf b}_{\rm BT}^{\rm H}{\bf w}_{\rm c}\vert^2}{{\beta}_{\rm BR}^2{\beta}_{\rm C}\vert{\bf u}^{\rm H}{\bf b}_{\rm RB,r}\vert^2\vert{\bf b}_{\rm BR}^{\rm H}{\bm\Theta}{\bf a}_{\rm clu,0}^{\rm r}({\bf a}_{\rm clu,0}^{\rm t})^{\rm H}{{\bm\Theta}}{\bf b}_{\rm BR}\vert^2\vert{\bf b}_{\rm BC}^{\rm H}{\bf w}_{\rm c}\vert^2\!+\!\Vert{\bf u}\Vert^2\sigma_{\rm R}^2},
\end{align}
where the vector set $\{{\bf a}_{{\rm tar},0}^{\rm t},{\bf a}_{{\rm tar},0}^{\rm r},{\bf a}_{{\rm clu},0}^{\rm t},{\bf a}_{{\rm clu},0}^{\rm r}\}$ is obtained by replacing $f_{n_{\rm t}},\forall n_{\rm t}\!\in\!\mathcal{N}_{\rm t}$ involved in $\{{\bf a}_{{\rm tar},n_{\rm t}}^{\rm t},{\bf a}_{{\rm tar},n_{\rm t}}^{\rm r},{\bf a}_{{c},n_{\rm t}}^{\rm t},{\bf a}_{c,n_{\rm t}}^{\rm r}\}$ with $f_{\rm ref}$. The vectors ${\bf b}_{\rm BR}\!\in\!\mathbb{C}^{M\times 1}$ and ${\bf b}_{\rm RB,r}\!\in\!\mathbb{C}^{N_{\rm r}\times 1}$ are respectively given by 
\begin{align}
    &[{\bf b}_{\rm BR}]_m\!=\!e^{-j2\pi\frac{f_{\rm ref}}{\rm c}((m_{\rm ele}\!-\!1)d_{\rm ele}\sin\theta_{\rm azi}^{\rm BR}\sin\theta_{\rm ele}^{\rm BR}\!+\!(m_{\rm azi}\!-\!1)d_{\rm azi}\sin\theta_{\rm azi}^{\rm BR}\cos\theta_{\rm ele}^{\rm BR})},\forall m\!\in\!\mathcal{M},\notag\\
    &[{\bf b}_{\rm RB,r}]_{n_{\rm r}}\!=\!e^{-j2\pi\frac{f_{\rm ref}}{\rm c}(n_{\rm r}-1)d_{\rm B}\sin\psi^{\rm BR}},\forall n_{\rm r}\!\in\!\mathcal{N}_{\rm r},
\end{align}
which are in essence the receive steering vectors at the conventional UPA-based RIS and ULA-based BS, respectively, whereas the vectors ${\bf b}_{\rm BT}\!\in\!\mathbb{C}^{N_
{\rm t}\times 1}$ and ${\bf b}_{\rm BC}\!\in\!\mathbb{C}^{N_
{\rm t}\times 1}$ are introduced to collect the exponential terms corresponding to FDA frequency offsets, i.e., $e^{j2\pi\frac{f_{n_{\rm t}}}{\rm c}D_{x}},x\!\in\!\{{\rm BR},{\rm RT},{\rm RC}\},\forall n_{\rm t}\!\in\!\mathcal{N}_{\rm t}$, which are defined as
\begin{align}
    [{\bf b}_{\rm BT}]_{n_{\rm t}}\!&=\!e^{-j2\pi(-\frac{f_{n_{\rm t}}}{\rm c}(2 D_{\rm BR}+2 D_{\rm RT})+\frac{f_{\rm ref}}{\rm c}(n_{\rm t}-1)d_{\rm B}\sin\theta_{\rm t}^{\rm BR})},\\
    [{\bf b}_{\rm BC}]_{n_{\rm t}}\!&=\!e^{-j2\pi(-\frac{f_{n_{\rm t}}}{\rm c}(2 D_{\rm BR}+2 D_{\rm RC})+\frac{f_{\rm ref}}{\rm c}(n_{\rm t}-1)d_{\rm B}\sin\theta_{\rm t}^{\rm BR})},~n_{\rm t}\!\in\!\mathcal{N}_{\rm t}.
\end{align}

Next,  in order to analytically demonstrate that the FDA  is  capable  of enhancing target  detection performance, given any RIS phase shift matrix ${\bm\Theta}$, we  aim  to  maximize  SCNR $\Gamma_{\rm SUST}$  by jointly optimizing the BS transmit beamforming vector ${\bf w}_{\rm c}$, the RIS phase shift matrix ${\bm\Theta}$, the radar receive equalizer ${\bf u}$ and the FDA transmit frequency offset increment $\Delta f$. Mathematically, the optimization problem is formulated as
    \begin{align}
        ({\text{P2}}):~  \underset{{\bf w}_{\rm c},{\bf u},\Delta f}
        {{\max}}~\Gamma_{\rm SUST},~~
        {\rm {s.t.}}~~
        \Vert {\bf w}_{\rm c}\Vert^2\leq P_{\rm B},
        ~~\Delta f\in(0,\Delta f_{\max}],,
    \end{align}
where $\Delta f_{\max}$ denotes the maximum allowable frequency offset increment. Problem (P2) is challenging to solve due to the non-convex fractional objective function and the highly coupled variables. Fortunately, by exploiting the its inherent characteristics, we have the following proposition.
    \begin{proposition}\label{prop-optwu}
        The optimal SCNR to problem (P2) for the SUST FDA-RIS-aided ISAC system is written as
        \begin{align}\label{SUST_FDA}
            {\Gamma}_{\rm SUST,max}^{\rm FDA}\!&=\!\frac{{\beta}_{\rm T}\vert p_{\rm tar}\vert^2N_{\rm r}P_{\rm B}}{\sigma_{\rm R}^2}\big(N_{\rm t}\!-\!\frac{\beta_{\rm C}P_{\rm B}\vert p_{\rm tar}\vert^2N_{\rm r}\big\vert\frac{\sin({N_{\rm t}}{2\pi\Delta f^{\rm opt}\Delta D}/{{\rm c}})}{\sin({2\pi\Delta f^{\rm opt}\Delta D}/{{\rm c}})}\big\vert^2}{\sigma_{\rm R}^2\!+\!\beta_{\rm C}P_{\rm B}N_{\rm t}\vert p_{\rm tar}\vert^2N_{\rm r}}\big),
        \end{align}
        with $\Delta f^{\rm opt}\!=\!\min\{\Delta f_{\max},\Delta f_0\}$. For its PA-RIS-aided counterpart, the optimal SCNR is expressed as
        \begin{align}\label{SUST_PA}
            {\Gamma}_{\rm SUST}^{\rm PA}\!=\!\frac{{\beta}_{\rm T}\vert p_{\rm tar}\vert^2N_{\rm r}P_{\rm B}}{\sigma_{\rm R}^2}\big(N_{\rm t}\!-\!\frac{\beta_{\rm C}P_{\rm B}\vert p_{\rm tar}\vert^2N_{\rm r}N_{\rm t}^2}{\sigma_{\rm R}^2\!+\!\beta_{\rm C}P_{\rm B}N_{\rm t}\vert p_{\rm tar}\vert^2N_{\rm r}}\big).
        \end{align}
        It follows from \eqref{SUST_FDA} and \eqref{SUST_PA} that ${\Gamma}_{\rm SUST}^{\rm FDA}\!\geq\!{\Gamma}_{\rm SUST}^{\rm PA}$ always holds with a maximum achievable SCNR increment of $\Delta\Gamma_{\max}\!=\!\frac{{\beta}_{\rm T}{\beta}_{\rm C}\vert p_{\rm tar}\vert^4N_{\rm r}^2P_{\rm B}^2(N_{\rm t}^2-\big\vert\frac{\sin({N_{\rm t}}{2\pi\Delta f^{\rm opt}\Delta D}/{{\rm c}})}{\sin({2\pi\Delta f^{\rm opt}\Delta D}/{{\rm c}})}\big\vert^2)}{\sigma_{\rm R}^2({\sigma_{\rm R}^2\!+\!\beta_{\rm C}P_{\rm B}N_{\rm t}\vert p_{\rm tar}\vert^2N_{\rm r}})}$.
    \end{proposition}
\begin{proof}
    Please refer to Appendix~\ref{app-optwu}.
\end{proof}

It is clear from Proposition~\ref{prop-optwu} that the FDA-aided BS always outperforms the PA-aided BS in terms of clutter suppression, thereby yielding a higher SCNR. Moreover, it can be readily inferred from $\Delta\Gamma_{\max}$ the corresponding SCNR increment exhibits linear growth with the BS transmit power and the number of BS receive antennas, i.e., $\Delta\Gamma_{\max}$ increases with asymptotically large $P_{\rm B}$ and $N_{\rm r}$ according to scaling orders of $\mathcal{O}(P_{\rm B})$ and $\mathcal{O}(N_{\rm r})$, respectively.

\section{Joint Active-Passive Beamforming and frequency offsets Optimization}\label{S3}

In this section, we firstly reformulate problem (P1-1) into a tractable parametric subtraction form by utilizing the FP-based technique, and then propose an efficient SADMM-SCA-based AO algorithm to jointly optimize the BS transmit beamforming vectors ${\bf w}_k$, the RIS phase shift matrix ${{\bm \Theta}}$, the FDA frequency offsets $\Delta f_{n_{\rm t}}$ and the radar receive equalizer ${\bf u}$ until convergence.

{\noindent\textit{A. FP-based Problem Reformulation}}

In this subsection, we intend to leverage the FP-based technique to rewrite the original problem (P1-1) into a parametric subtraction form, as demonstrated in the following Proposition.
\begin{proposition}\label{prop1}
    Problem (P1-1) can be equivalently reformulated as
    \begin{subequations}\label{P2}
        \begin{align}
            ({\text{P3}}):~  \underset{{{\mathcal{V}_{\rm P_3}}}}
            {{\max}}~~&\sum\limits_{k=1}^{K}\big( \log(1\!+\!w_k)\!-\!w_k\!+\!2\sqrt{1\!+\!w_k}\Re\{\alpha_k^*\tilde{\bf h}_k^{\rm H}{\bf w}_k\}\!-\!\vert\alpha_k\vert^2\big(\sum\limits_{k'=1}^{K}\vert\tilde{\bf h}^{\rm H}_k{\bf w}_{k'}\vert^2\!+\!\sigma_k^2\big) \big) \label{P2obj} \\
            {\rm {s.t.}}~~&\eqref{P1-1cons1},\eqref{P1-1cons2},\eqref{P1-1cons3},
        \end{align}
    \end{subequations}
    where $\mathcal{V}_{\rm P_3}=\{\alpha_{ k},w_{ k},{\bf w}_{ k},{\bm \Theta},\Delta f_{ n_{\rm t}},{\bf u}\}$.
\end{proposition}
\begin{proof}
    Please refer to Appendix~\ref{appA}.
\end{proof}
Unfortunately, the reformulated problem (P3) is still a complicated non-convex problem with highly coupled variables even though its objective function \eqref{P2obj} is more tractable as compared to the original form \eqref{P1-1obj}. In the following, we aim to develop an efficient SADMM-SCA-based AO algorithm to address problem (P3).

{\noindent\textit{B. Proposed SADMM-SCA-based AO Algorithm}}

In this subsection, we firstly decompose the optimization variables of problem (P3) into four blocks, i.e., ${\bf w}_k$, $\{\alpha_k,w_k,{\bf u}\}$, ${\bm\Theta}$ and $\Delta f_{n_{\rm t}}$. Then, we alternately and iteratively solve the corresponding four subproblems until convergence is reached.

{\noindent\textit{1) Optimization of ${\bf w}_k$}}

Given $\{\alpha_k,w_k,{\bm\Theta},\Delta f_{n_{\rm t}},{\bf u}\}$, the subproblem w.r.t. $\{{\bf w}_k\}_{k\in\mathcal{K}}$  can be expressed as 
\begin{subequations}\label{P21}
    \begin{align}
        ({\text{P3-1}}):~  \underset{\{{\bf w}_{ k}\}}
        {{\min}}~~
        &\sum\limits_{k=1}^{K} \big(\vert\alpha_k\vert^2\big(\sum\limits_{k'=1}^{K}\vert\tilde{\bf h}^{\rm H}_k{\bf w}_{k'}\vert^2\!+\!\sigma_k^2\big)\!-\!2\sqrt{1\!+\!w_k}\Re\{\alpha_k^*\tilde{\bf h}_k^{\rm H}{\bf w}_k\}\big) \label{P21obj} \\
        {\rm {s.t.}}~~&    \eqref{P1-1cons1},\eqref{P1-1cons2}.
    \end{align}
\end{subequations}
Clearly, the main difficulty for solving problem (P3-1) arises from the non-convex constraint \eqref{P1-1cons2}. To effectively address this issue, we explore the inherent characteristics of problem (P3-1) in the following proposition.
\begin{proposition}\label{prop-new-new-2}
    The constraint \eqref{P1-1cons1} in problem (P3-1) must be active at the optimum. 
\end{proposition}
\begin{proof}
    Please refer to Appendix~\ref{app-prop-new-new-2}.
\end{proof}

Motivated by Proposition 4, problem (P3-1) can be  much simplified. Specifically, based on the variable transformation ${\bf w}_k\!=\!\hat{\nu}_1 \hat{\bf w}_k$ and ${\alpha}_k\!=\!\frac{\hat{\alpha}_k}{\hat{\nu}_1}$ with ${\hat{\nu}_1}\!=\!\sqrt{\frac{P_{\rm B}}{\sum\limits_{k=1}^{K}\Vert\hat{\bf w}_k\Vert^2}}$, we can incorporate the constraint \eqref{P1-1cons1} into the objective function \eqref{P21obj} to reformulate it as
\begin{subequations}\label{P211}
    \begin{align}
        ({\text{P3-1-1}}):~  \underset{\{\hat{\bf w}_{ k}\}}
        {{\min}}~
        &\sum\limits_{k=1}^{K} \big(\vert\hat{\alpha}_k\vert^2\big(\sum\limits_{k'=1}^{K}\vert\tilde{\bf h}^{\rm H}_k\hat{\bf w}_{k'}\vert^2\!+\!\frac{\sigma_k^2}{P_{\rm B}}\sum\limits_{k'=1}^{K}\Vert\hat{\bf w}_{k'}\Vert^2\big)\!-\!2\sqrt{1\!+\!w_k}\Re\{\hat{\alpha}_k^*\tilde{\bf h}_k^{\rm H}\hat{\bf w}_k\}\big) \label{P211obj} \\
        {\rm {s.t.}}~~&\gamma_{\rm T}({P_{\rm B}\sum\limits_{c=1}^{C}\!\sum\limits_{k=1}^{K}\!\beta_{c}\vert{\bf u}^{\rm H}{\bf H}_{{\rm BC},c}\hat{\bf w}_k\vert^2\!+\!\Vert{\bf u}\Vert^2\sigma_{\rm R}^2\sum\limits_{k=1}^{K}\Vert\hat{\bf w}_k\Vert^2})\!\notag\\
        &-\!{P_{\rm B}\beta_{\rm T}\sum\limits_{k=1}^{K}\vert{\bf u}^{\rm H}{\bf H}_{\rm BT}\hat{\bf w}_k\vert^2}\!\leq\!0.\label{P211cons1}
    \end{align}
\end{subequations}

Obviously, problem (P3-1-1) is a quadratically constrained quadratic programming problem with a single constraint (QCQP-1), for which the strong duality strictly holds even though the constraint \eqref{P211cons1} is non-convex\cite[Chapter 5.2.4]{QCQP_cons1}. Therefore, by leveraging the Lagrangian dual theory, the optimal $\hat{\bf w}_k^{\rm opt}$ can be derived as
\begin{align}\label{optw}
    \hat{\bf w}_k^{\rm opt}\!&=\!\big(\sum\limits_{k'=1}^{K}\vert\hat{\alpha}_{k'}\vert^2(\tilde{\bf h}_{k'}\tilde{\bf h}_{k'}^{\rm H}\!+\!\frac{\sigma_{k'}^2}{P_{\rm B}}{\bf I}_{N_{\rm t}})\!+\!\mu_1{\bf P}_1\big)^{\dagger}\sqrt{1\!+\!w_k}\hat{\alpha}_k\tilde{\bf h}_k,
\end{align}
where ${\bf P}_1\!=\!\gamma_{\rm T}(P_{\rm B}\sum\limits_{c=1}^{C}\beta_{c}{\bf H}_{{\rm BC},c}^{\rm H}{\bf u}{\bf u}^{\rm H}{\bf H}_{{\rm BC},c}\!+\!\Vert{\bf u}\Vert^2\sigma_{\rm R}^2{\bf I}_{N_{\rm t}})\!-\!P_{\rm B}\beta_{\rm T}{\bf H}_{{\rm BT}}^{\rm H}{\bf u}{\bf u}^{\rm H}{\bf H}_{{\rm BT}}$. $\mu_1$ is the optimal dual variable associated with \eqref{P211cons1} which can be determined by the classical subgradient method\cite{SubMethod}. Finally, by recalling ${\bf w}_k\!=\!\hat{\nu}_1 \hat{\bf w}_k$, the optimal ${\bf w}_k$ can be obtained.

{\noindent\textit{2) Optimization of $\{\alpha_k, w_k, {\bf u}\}$} }

We clearly observe from problem (P3) that $\{\alpha_k,w_k\}$ are unconstrained optimization variables, whereas the equalizer ${\bf u}$ only appears in the SCNR constraint, all of which can be optimized in parallel. To be specific, given any $\{{\bf w}_k,{\bm\Theta},\Delta f_{n_{\rm t}}\}$, the subproblem w.r.t. $\{\alpha_k,w_k\}$ is formulated as 
\begin{align}
    ({\text{P3-2-1}}):~  \underset{\alpha_k,w_k}
    {{\max}}~~&  \log(1\!+\!w_k)\!-\!w_k\!+\!2\sqrt{1\!+\!w_k}\Re\{\alpha_k^*\tilde{\bf h}_k^{\rm H}{\bf w}_k\}\!-\!\vert\alpha_k\vert^2\big(\sum\limits_{k'=1}^{K}\vert\tilde{\bf h}^{\rm H}_k{\bf w}_{k'}\vert^2\!+\!\sigma_k^2\big),\forall k\!\in\!\mathcal{K}.
\end{align}
By applying the first-order optimality condition, the optimal $\alpha_k^{\rm opt}$ to problem (P3-2-1) is directly derived in the following closed-form, i.e.,
\begin{align}\label{a_opt}
    \alpha_k^{\rm opt}\!=\!\frac{\sqrt{1\!+\!w_k}\tilde{\bf h}_k^{\rm H}{\bf w}_k}{\sum\limits_{k'=1}^{K}\vert\tilde{\bf h}_{k}^{\rm H}{\bf w}_{k'}\vert^2\!+\!\sigma_k^2},\forall k\!\in\!\mathcal{K}.
\end{align}
Substituting $\alpha_k^{\rm opt}$ into the objective function of problem (P3-2-1) further yields
\begin{align}
    \log(1\!+\!w_k)\!-\!w_k\!+\!\frac{({1\!+\!w_k})\vert\tilde{\bf h}_k^{\rm H}{\bf w}_k\vert^2}{\sum\limits_{k'=1}^{K}\vert\tilde{\bf h}_{k}^{\rm H}{\bf w}_{k'}\vert^2\!+\!\sigma_k^2},\forall k\!\in\!\mathcal{K}.
\end{align}
Similar to deriving $\alpha_k^{\rm opt}$, the optimal $w_k$ is readily obtained as
\begin{align}
    w_k^{\rm opt}\!=\!\frac{\vert\tilde{\bf h}_k^{\rm H}{\bf w}_k\vert^2}{\sum\limits_{{k'\neq k,k'=1}}^{K}\vert\tilde{\bf h}_{k}^{\rm H}{\bf w}_{k'}\vert^2\!+\!\sigma_k^2},\forall k\!\in\!\mathcal{K}.
\end{align}

In  addition, it  follows  from problem (P3) again that  the optimization of ${\bf u}$ actually  belongs to  a feasibility-check problem, whose optimal solution aims to maximize the achievable SCNR, which renders more flexibility for the subsequent RIS phase shifts optimization, and thus enabling further enhancement of the sum rate.
As such, given $\{{\bf w}_k,{\bm\Theta},\Delta f_{n_{\rm t}}\}$,  the corresponding subproblem  w.r.t. ${\bf u}$ is expressed as
\begin{align}
    ({\text{P3-2-2}}):~  \underset{{\bf u}}
    {{\max}}~~&\frac{\beta_{\rm T}\sum\limits_{k=1}^{K}\vert{\bf u}^{\rm H}{\bf H}_{\rm BT}{\bf w}_k\vert^2}{\sum\limits_{c=1}^{C}\sum\limits_{k=1}^{K}\beta_c\vert{\bf u}^{\rm H}{\bf H}_{{\rm BC},c}{\bf w}_k\vert^2\!+\!\Vert{\bf u}\Vert^2\sigma_{\rm R}^2},
\end{align}
which is a standard generalized Rayleigh quotient problem, and its optimal solution can be derived in the following closed-form\cite{Rayquo}
\begin{align}
    {\bf u}^{\rm opt}\!=\!{\bm\mu}_{\rm eig}(\sum\limits_{k=1}^K{\bf H}_{\rm BT}{\bf w}_k{\bf w}_k^{\rm H}{\bf H}_{\rm BT}^{\rm H},\sum\limits_{c=1}^{C}\sum\limits_{k=1}^{K}\beta_c{\bf H}_{{\rm BC},c}{\bf w}_k{\bf w}_k^{\rm H}{\bf H}_{{\rm BC},c}^{\rm H}\!+\!\sigma_{\rm R}^2{\bf I}_{N_{\rm r}}),
\end{align}
where ${\bm\mu}_{\rm eig}({\bf A},{\bf B})$ denotes the normalized eigenvector of the matrix pencil $({\bf A},{\bf B})$ associated with the maximum eigenvalue.

{\noindent\textit{3) Optimization of ${\bm\Theta}$}}

In this subsection, we intend to optimize the RIS phase shift matrix ${\bm\Theta}$ while fixing the other variables. Inspired by the diagonal structure of ${\bm\Theta}$, we first introduce the RIS phase shift vector ${\bm\theta}\!=\!{\rm diag}\{{\bm\Theta}\}$ and utilize the equality ${\bf x}^{\rm H}{\bm\Theta}{\bf y}\!=\!{\bm\theta}^{\rm T}{{\rm diag}\{{\bf x}^{\rm H}\}{\bf y}}$ for any vectors ${\bf x},{\bf y}$ to re-express problem (P3) in terms of ${\bm \theta}$ as
\begin{subequations}\label{P231}
    \begin{align}
        ({\text{P3-3}}):~  \underset{{\bm\theta}}
        {{\min}}~~&f_{\rm RIS}({\bm\theta})\!\triangleq\!\sum\limits_{k=1}^{K}\!\big(\vert\alpha_k\vert^2\sum\limits_{k'=1}^{K}\vert{\bm\theta}^{\rm T}{\bf G}_{{\rm BU},k}{\bf w}_{k'}\vert^2\!-\!2\sqrt{1\!+\!w_k}\Re\{\alpha_k^*{\bm\theta}^{\rm T}{\bf G}_{{\rm BU},k}{\bf w}_k\}\big) \\
        {\rm {s.t.}}~~
        &{\rm Tr}({\bf T}_1{\bm\theta}{\bm\theta}^{\rm H}{\bf T}_2{\bm\theta}{\bm\theta}^{\rm H})\!\geq\!\gamma_{\rm T}\sum\limits_{c=1}^{C}\!{\rm Tr}({\bf R}_{1,c}{\bm\theta}{\bm\theta}^{\rm H}{\bf R}_{2,c}{\bm\theta}{\bm\theta}^{\rm H})\!+\!\gamma_{\rm T}\Vert{\bf u}\Vert^2\sigma_{\rm R}^2,\label{P231cons1}\\
        &\vert[{\bm\theta}]_m\vert\!=\!1,\forall m\!\in\!\mathcal{M},\label{P231cons2}
    \end{align}
\end{subequations}
where
\begin{subequations}\label{RISauxi}
    \begin{align}
        {\bf G}_{{\rm BU},k}\!&=\![{\rm diag}\{{\bf h}_{k,1}^{\rm H}\}{\bf h}_{{\rm BR,1}},\cdots,{\rm diag}\{{\bf h}_{k,N_{\rm t}}^{\rm H}\}{\bf h}_{{\rm BR,N_{\rm t}}}],\\
        {\bf T}_1\!&=\!\beta_{\rm T}{\rm diag}\{({\bf a}_{\rm tar,0}^{\rm r})^{\rm H}\}{\bf H}_{\rm RB,0}^{\rm H}{\bf u}{\bf u}^{\rm H}{\bf H}_{\rm RB,0}{\rm diag}\{{\bf a}_{\rm tar,0}^{\rm r}\},\\
        {\bf T}_2\!&=\!\sum\limits_{k=1}^{K}\!{\rm diag}\{{\bf a}_{\rm tar,0}^{\rm t}\}{\bf G}_{\rm BT}^{*}{\bf w}_k^{*}{\bf w}_k^{\rm T}{\bf G}_{\rm BT}^{\rm T}{\rm diag}\{({\bf a}_{\rm tar,0}^{\rm t})^{\rm H}\}.
    \end{align}
\end{subequations}
In \eqref{RISauxi}, ${\bf H}_{\rm RB,0}$ is obtained by replacing $f_{n_{\rm t}},\forall n_{\rm t}\!\in\!\mathcal{N}_{\rm t}$ involved in ${\bf H}_{{\rm RB},n_{\rm t}}$ with $f_{\rm ref}$, which is in essence the RIS-BS channel with the conventional PA-aided BS. The matrix ${\bf G}_{\rm BT}\!\in\!\mathbb{C}^{M\times N_{\rm t}}$ with $[{\bf G}_{\rm BT}]_{m,n_{\rm t}}\!=\![{\bf H}_{\rm BR}]_{m,n_{\rm t}}e^{-j2\pi\frac{f_{n_{\rm t}}}{{\rm c}}{{(D_{\rm BR}+2D_{\rm RT})}}},\forall m\!\in\!\mathcal{M},n_{\rm t}\!\in\!\mathcal{N}_{\rm t}$ is introduced to collect FDA frequency offsets. $\{{\bf R}_{1,c},{\bf R}_{2,c}\}$ are defined similarly to $\{{\bf T}_{1},{\bf T}_{2}\}$ with $\{\beta_{\rm T},{\bf a}_{\rm tar,0}^{\rm t},{\bf a}_{\rm tar,0}^{\rm r},{\bf G}_{\rm BT}\}$ replaced by $\{\beta_{c},{\bf a}_{c,0}^{\rm t},{\bf a}_{c,0}^{\rm r},{\bf G}_{{\rm BC},c}\}$. Notice that $\{{\bf a}_{c,0}^{\rm r},{\bf a}_{c,0}^{\rm t}\}$ are obtained by substituting $f_{n_{\rm t}}\!=\!f_{\rm ref}$ into $\{{\bf a}_{{c},n_{\rm t}}^{\rm r},{\bf a}_{{c},n_{\rm t}}^{\rm t}\}$ and ${\bf G}_{{\rm BC},c}\!\in\!\mathbb{C}^{M\times 1}$ is defined similarly to ${\bf G}_{{\rm BT}}$ by replacing $D_{\rm RT}$ with $D_{{\rm RC},c}$.

Unfortunately, problem (P3-3) is  still  difficult  to solve due to  the non-convex quartic constraint \eqref{P231cons1}. Hereafter, we adopt the majorization-minimization (MM) methodology to tackle this issue\cite{RIS_ISAC3}. Specifically, let us define ${\bm\psi}\!=\!{\rm vec}({{\bm\theta}{\bm\theta}^{\rm H}})$, then the left-hand-side of \eqref{P231cons1} is lower-bounded by
\begin{align}\label{LB}
    {\rm Tr}({\bf T}_{1}{\bm\theta}{\bm\theta}^{\rm H}{\bf T}_{2}{\bm\theta}{\bm\theta}^{\rm H})\!&\overset{(a)}{\geq}\!2\Re\{{\bm\psi}_{t}^{\rm H}({\bf T}_{1}^{\rm T}\!\otimes\!{\bf T}_{2}){\bm\psi}\}\!-\!{\bm\psi}_{t}^{\rm H}({\bf T}_{1}^{\rm T}\!\otimes\!{\bf T}_{2}){\bm\psi}_{t}\notag\\
    &={\bm\theta}^{\rm H}({\bf T}_1{\bm\theta}_{t}{\bm\theta}_{t}^{\rm H}{\bf T}_2\!+\!{\bf T}_2{\bm\theta}_{t}{\bm\theta}_{t}^{\rm H}{\bf T}_1){\bm\theta}\!-\!{\bm\theta}_{t}^{\rm H}{\bf T}_1{\bm\theta}_{t}{\bm\theta}_{t}^{\rm H}{\bf T}_2{\bm\theta}_{t},
\end{align}
where ${\bm\theta}_{t}$ denotes the phase shift vector obtained at the $t$-th MM iteration and ${\bm\psi}_{t}\!=\!{\rm vec}({\bm\theta}_{t}{\bm\theta}_{t}^{\rm H})$. The inequality $(a)$ holds based on ${\bm\psi}\!=\!{\rm vec}({\bm\theta}{\bm\theta}^{\rm H})$ and the first-order Taylor expansion of the function ${\rm Tr}({\bf T}_{1}{\bm\theta}{\bm\theta}^{\rm H}{\bf T}_{2}{\bm\theta}{\bm\theta}^{\rm H})$ at the point ${\bm\psi}_{t}$. Similarly, the term ${\rm Tr}({\bf R}_{1,c}{\bm\theta}{\bm\theta}^{\rm H}{\bf R}_{2,c}{\bm\theta}{\bm\theta}^{\rm H})$ in the right-hand-side of \eqref{P231cons1} is upper-bounded by 
\begin{align}\label{UB}
    {\rm Tr}({\bf R}_{1,c}{\bm\theta}{\bm\theta}^{\rm H}{\bf R}_{2,c}{\bm\theta}{\bm\theta}^{\rm H})\!&\overset{(b)}{\leq}\!\lambda_{\max}({\bf R}_{1,{c}}^{\rm T}\!\otimes\!{\bf R}_{2,c}){\bm\psi}^{\rm H}{\bm\psi}\!+\!2\Re\{{\bm \psi}^{\rm H}({\bf R}_{1,{c}}^{\rm T}\!\otimes\!{\bf R}_{2,c}\!-\!\lambda_{\max}({\bf R}_{1,{c}}^{\rm T}\!\otimes\!{\bf R}_{2,c}){\bf I}_{M^2}){\bm \psi}_{\rm t}\}\notag\\
    &+\!{\bm\psi}_{t}^{\rm H}(\lambda_{\max}({\bf R}_{1,{c}}^{\rm T}\!\otimes\!{\bf R}_{2,c}){\bf I}_{M^2}\!-\!{ {\bf R}_{1,{c}}^{\rm T}\!\otimes\!{\bf R}_{2,c}}){\bm\psi}_{t}\notag\\
    &\overset{(c)}{=}\!{\bm\theta}^{\rm H}({\bf R}_{1,c}{\bm\theta}_t{\bm\theta}_{t}^{\rm H}{\bf R}_{2,c}\!+\!{\bf R}_{2,c}{\bm\theta}_t{\bm\theta}_{t}^{\rm H}{\bf R}_{1,c}\!-\!2\lambda_{\max}({\bf R}_{1,{c}}^{\rm T}\!\otimes\!{\bf R}_{2,c}){\bm\theta}_t{\bm\theta}_t^{\rm H}){\bm\theta}\notag\\
    &+\!2M^2\lambda_{\max}({\bf R}_{1,{c}}^{\rm T}\!\otimes\!{\bf R}_{2,c})\!-\!{\bm\psi}_t^{\rm H}{\bf R}_{1,{c}}^{\rm T}\!\otimes\!{\bf R}_{2,c}{\bm\psi}_t.
\end{align}
The inequality $(b)$ holds due to ${\bm\psi}\!=\!{\rm vec}({\bm\theta}{\bm\theta}^{\rm H})$ and the MM inequality ${\bf x}^{\rm H}{\bf X x}\!\leq\!{\bf x}^{\rm H}{\bf Y x}\!+\!2\Re\{{\bf x}^{\rm H}({\bf Y}\!-\!{\bf X}){\bf x}_{t}\}\!+\!{\bf x}_{t}^{\rm H}({\bf X}\!-\!{\bf Y}){\bf x}_{t}$ for any two Hermitian matrices ${\bf X}$ and ${\bf Y}$ satisfying ${\bf Y}\!\succeq\!{\bf X}$, and the equality $(c)$ holds due to ${\bm\psi}^{\rm H}{\bm\psi}\!=\!M^2$ and ${\bm\psi}\!=\!{\rm vec}({\bm\theta}{\bm\theta}^{\rm H})$.

Based on the approximations in \eqref{LB} and \eqref{UB}, a tractable upper-bound approximation of the original problem (P3-3) is written as 
    \begin{align}\label{P232}
        ({\text{P3-3-1}}):~  \underset{{\bm\theta}}
        {{\min}}~f_{\rm RIS}({\bm\theta}),~~
        {\rm {s.t.}}~
        {\bm\theta}^{\rm H}{\bf G}_{\rm RIS}{\bm\theta}\!+\!r_{\rm RIS}\!\leq\!0,\,
        \eqref{P231cons2},
    \end{align}
where ${\bf G}_{\rm RIS}\!=\!\gamma_{\rm T}\!\sum\limits_{c=1}^{C}({\bf R}_{1,c}{\bm\theta}_t{\bm\theta}_{t}^{\rm H}{\bf R}_{2,c}\!+\!{\bf R}_{2,c}{\bm\theta}_t{\bm\theta}_{t}^{\rm H}{\bf R}_{1,c}\!-\!2\lambda_{\max}({\bf R}_{1,{c}}^{\rm T}\!\otimes\!{\bf R}_{2,c}){\bm\theta}_t{\bm\theta}_t^{\rm H})\!-\!{\bf T}_1{\bm\theta}_{t}{\bm\theta}_{t}^{\rm H}{\bf T}_2\!-\!{\bf T}_2{\bm\theta}_{t}{\bm\theta}_{t}^{\rm H}{\bf T}_1$, and $r_{\rm RIS}\!=\!\gamma_{\rm T}\!\sum\limits_{c=1}^{C}(2M^2\lambda_{\max}({\bf R}_{1,{c}}^{\rm T}\otimes{\bf R}_{2,c})\!-\!\psi_t^{\rm H}{\bf R}_{1,{c}}^{\rm T}\!\otimes\!{\bf R}_{2,c}\psi_t)\!+\!{\bm\theta}_{t}^{\rm H}{\bf T}_1{\bm\theta}_{t}{\bm\theta}_{t}^{\rm H}{\bf T}_2{\bm\theta}_{t}\!+\!\gamma_{\rm T}\Vert{\bf u}\Vert^2\sigma_{\rm R}^2$. The challenge left for solving problem (P3-3-1) is the intractable unit-modulus constraint \eqref{P231cons2}. To tackle this challenge, we consider applying the SADMM method with faster convergence rate than the conventional ADMM method to solve problem (P3-3-1) in an iterative manner\cite{symADMM}. Specifically, by introducing an auxiliary variable ${\bm \phi}$, problem (P3-3-1) can be rewritten as
\begin{subequations}\label{P233}
    \begin{align}
        ({\text{P3-3-2}}):~  \underset{{\bm \theta},{\bm \phi}}
        {{\min}}~&f_{\rm RIS}({\bm\theta}) \\
        {\rm {s.t.}}~~
        &{\bm\theta}^{\rm H}{\bf G}_{\rm RIS}{\bm\theta}\!+\!r_{\rm RIS}\!\leq\!0,~{\bm \phi}\!=\!{\bm\theta},~~\vert[{\bm \phi}]_m\vert\!=\!1,~\forall m\!\in\!\mathcal{M}.\label{P233cons3}
    \end{align}
\end{subequations}

Under the SADMM framework, the scaled augmented Lagrangian function of problem (P3-3-2) is given by
\begin{align}\label{Lag}
    \mathcal{L}({\bm\phi},{\bm \theta},{\bm\rho})\!=\!f_{\rm RIS}({\bm\theta})\!+\!\frac{\mu_{\rm pen}}{2}\big\Vert{\bm\theta}\!-\!{\bm\phi}\!+\!\frac{\bm\rho}{\mu_{\rm pen}}\big\Vert^2,
\end{align}
where ${\bm \rho}$ and $\mu_{\rm pen}$ are the dual variable and the penalty factor, respectively. Then, the SADMM update rules  are presented as follows:
\begin{subequations}\label{admmite}
\begin{align}
    {\bm\phi}^{(j)}\!&:=\!\underset{{\bm\phi}}{\arg\min}\mathcal{L}({\bm\phi},{\bm \theta}^{(j-1)},{\bm\rho}^{(j-1)}),\label{admm1}\\
    {\bm\rho}^{(j-\frac{1}{2})}\!&:=\!{\bm\rho}^{(j-1)}\!+\!r_1\mu_{\rm pen}({\bm \theta}^{(j-1)}\!-\!{\bm\phi}^{(j)}),\\
    {\bm\theta}^{(j)}\!&:=\!\underset{{\bm\theta}}{\arg\min}\mathcal{L}({\bm\phi}^{(j)},{\bm \theta},{\bm\rho}^{(j-\frac{1}{2})}),\label{admm2}\\{\bm\rho}^{(j)}\!&:=\!{\bm\rho}^{(j-\frac{1}{2})}\!+\!r_2\mu_{\rm pen}({\bm \theta}^{(j)}\!-\!{\bm\phi}^{(j)}),
\end{align}
\end{subequations}
where $j$ is the SADMM iteration index, $r_1$ and $r_2$ are independent stepsizes restricted by $\mathcal{D}\!=\!\{(r_1,r_2)|r_1\!\in\!(-1,1),r_2\!\in\!(0,\frac{1\!+\!\sqrt{5}}{2}),r_1\!+\!r_2\!>\!0,\vert r_1\vert\!<\!1\!+\!r_2\!-\!r_2^2\}$\cite{symADMM}.

On the one hand, by omitting the constant irrelevant to ${\bm\phi}$, we can express the subproblem \eqref{admm1} w.r.t. ${\bm\phi}$ as
\begin{align}\label{proadmm1}
    \underset{{\bm \phi}}
    {{\min}}~\big\Vert{\bm\theta}^{(j-1)}\!-\!{\bm\phi}\!+\!\frac{{\bm\rho}^{(j-1)}}{\mu_{\rm pen}}\big\Vert^2, ~~~{\rm {s.t.}}~\vert[{\bm\phi}]_m\vert\!=\!1,\forall m\!\in\!\mathcal{M},
\end{align}
whose optimal solution is readily derived as
\begin{align}
    {\bm\phi}^{(j)}\!=\!e^{j\angle({\bm\theta}^{(j-1)}+\frac{{\bm\rho}^{(j-1)}}{\mu_{\rm pen}})}.
\end{align}

On the other hand, the subproblem \eqref{admm2} w.r.t. ${\bm\theta}$ can be reformulated as a QCQP-1 problem, which is given by 
    \begin{align}\label{P245}
        \underset{{\bm \theta}}
        {{\min}}~f_{\rm RIS}({\bm\theta})\!+\!\frac{\mu_{\rm pen}}{2}\big\Vert{\bm\theta}\!-\!{\bm\phi}^{(j)}\!+\!\frac{{\bm\rho}^{(j-\frac{1}{2})}}{\mu_{\rm pen}}\big\Vert^2, ~~~{\rm {s.t.}}~{\bm\theta}^{\rm H}{\bf G}_{\rm RIS}{\bm\theta}\!+\!r_{\rm RIS}\!\leq\!0.
    \end{align}
Naturally, the strong duality holds for problem \eqref{P245}. Therefore, similar to solving problem (P3-1-1), the optimal ${\bm\theta}^{(j)}$ can be obtained as
\begin{align}
    {\bm\theta}^{(j)}\!&=\!\big(\sum\limits_{k=1}^{K}\!\vert\alpha_k\vert^2\sum\limits_{k'=1}^{K}\!({\bf G}_{{\rm BU},k}^*{\bf w}_{k'}^*{\bf w}_{k'}^{\rm T}{\bf G}_{{\rm BU},k}^{\rm T})\!+\!\frac{\mu_{\rm pen}}{2}{\bf I}_M\!+\!\mu_2{\bf G}_{\rm RIS}\big)^{\dagger}\notag\\
    &\times\!\big(\frac{\mu_{\rm pen}}{2}{\bm\phi}^{(j)}\!-\!\frac{{\bm\rho}^{(j-\frac{1}{2})}}{2}\!+\sum\limits_{k=1}^{K}\!\sqrt{1\!+\!w_k}\alpha_k{\bf G}_{{\rm BU},k}^*{\bf w}_k^*\big),
\end{align}
where $\mu_2$ is the optimal dual variable associated with the constraint in \eqref{P245} and further determined by the subgradient method\cite{SubMethod}.  Overall, based on the SADMM iterations shown in \eqref{admmite}, we ultimately obtain the locally optimal ${\bm\theta}^{\rm opt}$ to the original subproblem (P3-3).

{\noindent\textit{5) Optimization of $\Delta f_{ n_{\rm t}}$}}

In this subsection, we aim to optimize the frequency offsets $\{\Delta f_{ n_{\rm t}}\}_{n_{\rm t}\in\mathcal{N}_{\rm t}}$ with the other variables being fixed. Since $\{\Delta f_{ n_{\rm t}}\}_{n_{\rm t}\in\mathcal{N}_{\rm t}}$ are implicitly involved in both the objective function and SCNR constraint of problem (P3), we introduce the vectors ${\bf f}_{\rm BR}\!\in\!\mathbb{C}^{N_{\rm t}\times 1}$, ${\bf f}_{{\rm BU},k}\!\in\!\mathbb{C}^{N_{\rm t}\times 1}$, ${\bf f}_{\rm BT}\!\in\!\mathbb{C}^{N_{\rm t}\times 1}$, and ${\bf f}_{{\rm BC},c}\!\in\!\mathbb{C}^{N_{\rm t}\times 1}$ to re-express problem (P3), which are respectively given by
\begin{align}\label{fvectors}
    [{\bf f}_{\rm BR}]_{n_{\rm t}}\!&=\!e^{j\frac{2\pi}{\rm c}D_{\rm BR}\Delta f_{n_{\rm t}}},\,[{\bf f}_{{\rm BU},k}]_{n_{\rm t}}\!=\!e^{j\frac{2\pi}{\rm c}(D_{\rm BR}+D_{{\rm RU},k})\Delta f_{n_{\rm t}}},\notag\\
    [{\bf f}_{{\rm BT}}]_{n_{\rm t}}\!&=\!e^{j\frac{2\pi}{\rm c}(2D_{\rm BR}+2D_{\rm RT})\Delta f_{n_{\rm t}}},\,[{\bf f}_{{\rm BC},c}]_{n_{\rm t}}\!=\!e^{j\frac{2\pi}{\rm c}(2D_{\rm BR}+2D_{{\rm RC},c})\Delta f_{n_{\rm t}}},\,\forall n_{\rm t}\!\in\!N_{\rm t}.
\end{align}

Based on \eqref{fvectors}, the subproblem w.r.t. the frequency offsets $\{\Delta f_{ n_{\rm t}}\}_{n_{\rm t}\in\mathcal{N}_{\rm t}}$ is then expressed as
\begin{subequations}\label{P25}
    \begin{align}
        \!\!({\text{P3-4}}):~  \underset{\{{\Delta f_{ n_{\rm t}}}\}}
        {{\min}}~&\sum\limits_{k=1}^{K}\!\Big(\vert\alpha_k\vert^2\!\sum\limits_{k'=1}^{K}\!\vert{{\bf w}_{k'}^{\rm H}}({{\bf g}_{{\rm BU},k}^{\rm LoS}\!\circ\!{\bf f}_{{\rm BU}_k}}\!+\!{{\bf g}_{{\rm BU},k}^{\rm NLoS}\!\circ\!{\bf f}_{{\rm BR}}})\vert^2
        \notag\\
        &-\!2\sqrt{1\!+\!w_k}\Re\{\alpha_k{\bf w}_k^{\rm H}({{\bf g}_{{\rm BU},k}^{\rm LoS}\!\circ\!{\bf f}_{{\rm BU},k}}\!+\!{{\bf g}_{{\rm BU},k}^{\rm NLoS}\!\circ\!{\bf f}_{{\rm BR}}})\}\Big)\label{P25obj}\\
        {\rm {s.t.}}~~&\gamma_{\rm T}\sum\limits_{c=1}^{C}\sum\limits_{k=1}^{K}\big\vert{\bf w}_{k}^{\rm H}({\bf g}_{{\rm BC},c}\!\circ\!{\bf f}_{{\rm BC},c})\big\vert^2\!-\!\sum\limits_{k=1}^{K}\big\vert{\bf w}_{k}^{\rm H}({\bf g}_{\rm BT}\!\circ\!{\bf f}_{\rm BT})\big\vert^2\!+\!\gamma_{\rm T}\Vert{\bf u}\Vert^2\sigma_{\rm R}^2\!\leq\! 0,\label{P25cons}\\
        &\eqref{P1cons4},
    \end{align}
\end{subequations}
where the vectors $\{{\bf g}_{{\rm BU},k}^{\rm LoS},{\bf g}_{{\rm BU},k}^{\rm NLoS},{\bf g}_{{\rm BC},c},{\bf g}_{{\rm BT}}\}$ are given by
\begin{align}
    \!\!\!\!\!({\bf g}_{{\rm BU},k}^{\rm LoS})^{\rm H}\!&=\!\sqrt{\frac{\kappa\beta_{{\rm RU},k}}{\kappa+1}}({\bf h}_{k,0}^{\rm LoS})^{\rm H}{\bm\Theta}{\bf G}_{\rm BR},~({\bf g}_{{\rm BU},k}^{\rm NLoS})^{\rm H}\!=\!\sqrt{\frac{\beta_{{\rm RU},k}}{\kappa+1}}{\bm\theta}^{\rm T}({\bf G}_{{\rm RU},k}^{\rm NLoS}\!\circ\!{\bf G}_{\rm BR}),\notag\\
    {\bf g}_{{\rm BC},c}^{\rm H}\!&=\!\sqrt{\beta_c}{\bf u}^{\rm H}{\bf H}_{\rm RB,0}{\bm\Theta}{\bf a}_{c,0}^{\rm r}({\bf a}_{c,0}^{\rm t})^{\rm H}{\bm\Theta}{\bf G}_{\rm BR},~{\bf g}_{{\rm BT}}^{\rm H}\!=\!\sqrt{\beta_{\rm T}}{\bf u}^{\rm H}{\bf H}_{\rm RB,0}{\bm\Theta}{\bf a}_{\rm tar,0}^{\rm r}({\bf a}_{\rm tar,0}^{\rm t})^{\rm H}{\bm\Theta}{\bf G}_{\rm BR}.
\end{align}
Note that $\{{\bf h}_{k,0}^{\rm LoS},{\bf G}_{\rm BR}\}$ are obtained by substituting $f_{n_{\rm t}}\!=\!f_{\rm ref},\forall n_{\rm t}\!\in\!\mathcal{N}_{\rm t}$ into $\{{\bf h}_{k,n_{\rm t}}^{\rm LoS},{\bf H}_{\rm BR}\}$, and ${\bf G}_{{\rm RU},k}^{\rm NLoS}\!=\![{\bf h}_{k,1}^{\rm NLoS},\cdots,{\bf h}_{k,N_{\rm t}}^{\rm NLoS}]$. Although $\{\Delta f_{ n_{\rm t}}\}_{n_{\rm t}\in\mathcal{N}_{\rm t}}$ have been shown in a more explicit manner in problem (P3-4), they are still embedded in the exponential parts of $\{{\bf f}_{\rm BR},{\bf f}_{{\rm BU},k},{\bf f}_{\rm BT},{\bf f}_{{\rm BC},c}\}$ and tightly coupled with each other. Hereafter, we aim to tackle these exponential terms and then optimize each $\Delta f_{n_{\rm t}}$ in an element-wise manner. Specifically, by leveraging the following equalities
\begin{align}
    \vert{\bf x}^{\rm H}{\bf y}\vert^2\!&=\!\sum\limits_{p=1}^{N}\sum\limits_{q=1}^{N}\big|[{\bf x}]_{p}\big|\big|[{\bf y}]_{p}\big|\big|[{\bf x}]_{q}\big|\big|[{\bf y}]_{q}\big|\cos\big(\angle[{\bf x}]_{q}\!-\!\angle[{\bf x}]_{p}\!+\!\angle[{\bf y}]_{p}\!-\!\angle[{\bf y}]_{q}\big),\\
    \Re\{{\bf x}^{\rm H}{\bf y}\}\!&=\!\sum\limits_{p=1}^{N}\big|[{\bf x}]_{p}\big|\big|[{\bf y}]_{p}\big|\cos\big(\angle[{\bf y}]_{p}\!-\!\angle[{\bf x}]_{p}\big),
\end{align}
for any two vectors ${\bf x},{\bf y}\!\in\!\mathbb{C}^{N\times 1}$, problem (P3-4) can be rewritten in terms of each $\Delta f_{n_{\rm t}}$ as
\begin{subequations}\label{P251}
    \begin{align}
        \!\!({\text{P3-4-1}}):~  \underset{{\Delta f_{ n_{\rm t}}}}
        {{\min}}~&\sum\limits_{k=1}^{K}\!\Big(\sum\limits_{k'=1}^{K}\sum\limits_{i=1}^{3}g_{i,k,{k'}}(\Delta f_{n_{\rm t}})\!+\!\sum\limits_{j=1}^{2}g_{j,k}(\Delta f_{n_{\rm t}})\Big)\label{P251obj} \\
        {\rm {s.t.}}~~&\gamma_{\rm T}\sum\limits_{c=1}^{C}\sum\limits_{k=1}^{K} g_{c,{k}}(\Delta f_{n_{\rm t}})\!+\!\sum\limits_{k=1}^{K}g_{{\rm T},{k}}(\Delta f_{n_{\rm t}})\!+\!\gamma_{\rm T}^{\rm cons}\!\leq\!0,~\eqref{P1cons4}.\label{P251cons}
    \end{align}
\end{subequations}
In \eqref{P251obj}, the functions $g_{1,k,{k'}}(\Delta f_{n_{\rm t}})$, $g_{2,k,{k'}}(\Delta f_{n_{\rm t}})$ and $g_{3,k,{k'}}(\Delta f_{n_{\rm t}})$ represent reformulations of $\vert\alpha_k\vert^2\vert{{\bf w}_{k'}^{\rm H}}({{\bf g}_{{\rm BU},k}^{\rm LoS}\!\!\circ\!{\bf f}_{{\rm BU},k}})\vert^2$, $\vert\alpha_k\vert^2\vert{{\bf w}_{k'}^{\rm H}}({{\bf g}_{{\rm BU},k}^{\rm NLoS}\!\!\circ\!{\bf f}_{{\rm BR}}})\vert^2$ and $\vert\alpha_k\vert^2\Re\{{\bf w}_{k'}^{\rm H}({{\bf g}_{{\rm BU},k}^{\rm LoS}\!\!\circ\!{\bf f}_{{\rm BU},k}})({{\bf g}_{{\rm BU},k}^{\rm NLoS}\!\!\circ\!{\bf f}_{{\rm BR}}})^{\rm H}{\bf w}_{k'}\}$, respectively, with all  constants irrelevant to $\Delta f_{n_{\rm t}}$ omitted. $g_{1,k}(\Delta f_{n_{\rm t}})$ and $g_{2,k}(\Delta f_{n_{\rm t}})$ are similarly obtained from $-\sqrt{1\!+\!w_k}\Re\{\alpha_k{\bf w}_k^{\rm H}({{\bf g}_{{\rm BU},k}^{\rm LoS}\!\circ\!{\bf f}_{{\rm BU},k}})\}$ and $-\sqrt{1\!+\!w_k}\Re\{\alpha_k{\bf w}_k^{\rm H}({{\bf g}_{{\rm BU},k}^{\rm NLoS}\!\circ\!{\bf f}_{{\rm BR}}})\}$, respectively. Moreover, in the SCNR constraint in \eqref{P251cons}, $g_{c,{k}}(\Delta f_{n_{\rm t}})$ and $g_{{\rm T},{k}}(\Delta f_{n_{\rm t}})$ denote the $\Delta f_{n_{\rm t}}$-related  reformulations of $\big\vert{\bf w}_{k}^{\rm H}({\bf g}_{{\rm BC},c}\!\circ\!{\bf f}_{{\rm BC},c})\big\vert^2$ and $-\big\vert{\bf w}_{k}^{\rm H}({\bf g}_{\rm BT}\!\circ\!{\bf f}_{\rm BT})\big\vert^2$, respectively, and $\gamma_{\rm T}^{\rm cons}$ represents the remaining constant irrelevant to $\Delta f_{n_{\rm t}}$.  It is worth noting that all of the above functions $\{g_{i,k,{k'}}(\Delta f_{n_{\rm t}}),g_{j,k}(\Delta f_{n_{\rm t}}),g_{c,{k}}(\Delta f_{n_{\rm t}}),g_{{\rm T},{k}}(\Delta f_{n_{\rm t}}),i\!=\!1,2,3,j\!=\!1,2\}$ take the common form $g_0(\Delta f_{n_{\rm t}})\!=\!\sum\limits_{l}\tilde{g}_{l}(\Delta f_{n_{\rm t}})$ with $\tilde{g}_{l}(\Delta f_{n_{\rm t}})\!\triangleq\!\xi_{l}{\cos}(\eta_{l}\Delta f_{n_{\rm t}}+\rho_{l})$ and $l$ being the summation index, where the values of $\{\xi_{l},\eta_{l},\rho_{l}\}$ associated with different functions are shown in our {\textbf{Supplementary Material \cite{todo}}}.

\iffalse
\begin{align}
    \gamma_{\rm T}^{\rm cons}\!&=\!\sum\limits_{k\in\mathcal{K},{p,q\in\mathcal{N_{\rm t}}/\{n_{\rm t}\}}}\big(\sum\limits_{c=1}^{C}\gamma_{\rm T}[{\bf w}_{k}]_{p}^{*}[{\bf g}_{{\rm BC},c}]_{p}[{\bf g}_{{\rm BC},c}]_q^*[{\bf w}_{k}]_q\!-\![{\bf w}_{k}]_{p}^{*}[{\bf g}_{{\rm BT}}]_{p}[{\bf g}_{{\rm BT}}]_q^*[{\bf w}_{k}]_q\big)\notag\\
    &+\!\sum\limits_{k=1}^{K}\big(\sum\limits_{c=1}^{C}\gamma_{\rm T}\big\vert[{\bf w}_{k}]_{n_{\rm t}}^{*}[{\bf g}_{{\rm BC},c}]_{n_{\rm t}}\big\vert^2\!-\!\big\vert[{\bf w}_{k}]_{n_{\rm t}}^{*}[{\bf g}_{{\rm BT}}]_{n_{\rm t}}\big\vert^2\big)\!+\!\gamma_{\rm T}\Vert{\bf u}\Vert^2\sigma_{\rm R}^2.
\end{align}
\fi

The primary challenge for solving problem (P3-4-1) arises from handling the  non-convex cosine-form function $\tilde{g}_{l}(\Delta f_{n_{\rm t}})$. To resolve this issue, we consider leveraging the SCA methodology to construct a  tractable surrogate function of $\tilde{g}_{l}(\Delta f_{n_{\rm t}})$. Specifically, by referring to \cite{FDA_LPF1} and leveraging the MM philosophy, a convex quadratic upper bound of $\tilde{g}_{l}(\Delta f_{n_{\rm t}})$ at the $j$-th iteration point $\Delta f_{n_{\rm t}}^{(j)}$ is given by $\hat{g}_{l}(\Delta f_{n_{\rm t}})\!=\!\hat{\xi}_{l}^{(j)}(\Delta f_{n_{\rm t}}\!-\!\hat{\eta}_{l}^{(j)})^2\!+\!\hat{\rho}_{l}^{(j)}$, where the values of $\{\hat{\xi}_{l}^{(j)},\hat{\eta}_{l}^{(j)},\hat{\rho}_{l}^{(j)}\}$ are shown in the equations (24) and (25) in \cite{FDA_LPF1} corresponding to the two cases of $\tilde{g}_{l}'(\Delta f_{n_{\rm t}}^{(j)})\!\neq\!0$ and $\tilde{g}_{l}'(\Delta f_{n_{\rm t}}^{(j)})\!=\!0$, respectively. As such, a tractable upper bound approximation of problem (P3-4-1) is given by
\begin{subequations}\label{P242}
    \begin{align}
        \!\!({\text{P3-4-2}}):~  \underset{{\Delta f_{ n_{\rm t}}}}
        {{\min}}~&\hat{f}_{\rm obj}(\Delta f_{n_{\rm t}})\!\triangleq\!\sum\limits_{k=1}^{K}\!\Big(\sum\limits_{k'=1}^{K}\sum\limits_{i=1}^{3}\hat{g}_{i,k,{k'}}(\Delta f_{n_{\rm t}})\!+\!\sum\limits_{j=1}^{2}\hat{g}_{j,k}(\Delta f_{n_{\rm t}})\Big)\label{P252obj} \\
        {\rm {s.t.}}~~&\gamma_{\rm T}\sum\limits_{c=1}^{C}\sum\limits_{k=1}^{K}   \hat{g}_{c,{k}}(\Delta f_{n_{\rm t}})\!+\!\sum\limits_{k=1}^{K}\hat{g}_{{\rm T},{k}}(\Delta f_{n_{\rm t}})\!+\!\gamma_{\rm T}^{\rm cons}\!\leq\!0,~\eqref{P1cons4},\label{P252cons}
    \end{align}
\end{subequations}
where $\hat{g}_{i,k,{k'}}(\Delta f_{n_{\rm t}}),\hat{g}_{j,k}(\Delta f_{n_{\rm t}}),\hat{g}_{c,{k}}(\Delta f_{n_{\rm t}}),\hat{g}_{{\rm T},{k}}(\Delta f_{n_{\rm t}}),i\!=\!1,2,3,j\!=\!1,2$ denote the convex quadratic upper bound approximations of $g_{i,k,{k'}}(\Delta f_{n_{\rm t}}),g_{j,k}(\Delta f_{n_{\rm t}}),g_{c,{k}}(\Delta f_{n_{\rm t}}),g_{{\rm T},{k}}(\Delta f_{n_{\rm t}}),i\!=\!1,2,3,j\!=\!1,2$ at the $j$-th iteration point $\Delta f_{n_{\rm t}}^{(j)}$, respectively, with their specific mathematical expressions shown in our {\textbf{Supplementary Material \cite{todo}}}. Note that both the objective function \eqref{P252obj} and the SCNR constraint in \eqref{P252cons} are the summations of a series of quadratic functions w.r.t. $\Delta f_{n_{\rm t}}$, and thus can be equivalently rewritten as $\hat{f}_{\rm obj}(\Delta f_{n_{\rm t}})\!=\!\hat{d}_1\Delta f_{n_{\rm t}}^2\!+\!\hat{d}_2\Delta f_{n_{\rm t}}\!+\!\hat{d}_3$ and $\tilde{d}_1\Delta f_{n_{\rm t}}^2\!+\!\tilde{d}_2\Delta f_{n_{\rm t}}\!+\!\tilde{d}_3\!\leq\!0$, respectively, with the definitions of  $\{\hat{d}_1,\hat{d}_2,\hat{d}_3\}$ and $\{\tilde{d}_1,\tilde{d}_2,\tilde{d}_3\}$ also presented in our {\textbf{Supplementary Material \cite{todo}}}. Consequently, the optimal $\Delta f_{n_{\rm t}}^{\rm opt}$ is determined as
\begin{equation}
    \Delta f_{n_{\rm t}}^{\rm opt}\!=\! \begin{cases} {\max}\{\frac{-\tilde{d}_2-\sqrt{\tilde{d}_2^2-4\tilde{d}_1\tilde{d}_2}}{2\tilde{d}_1},0\}, & \!\!\!\! -\frac{\hat{d}_2}{2\hat{d}_1}\!<\! {\max}\{\frac{\tilde{d}_1\tilde{d}_2-\sqrt{-\tilde{d}_1\tilde{d}_3}}{\tilde{d}_1},0\} \\ -\frac{\hat{d}_2}{2\hat{d}_1}, &\!\!\!\!\!\!\!\!\!\! {\max}\{\frac{-\tilde{d}_2-\sqrt{\tilde{d}_2^2-4\tilde{d}_1\tilde{d}_2}}{2\tilde{d}_1},0\} \!\leq\! -\frac{\hat{d}_2}{2\hat{d}_1} \!<\! {\min}\{\frac{-\tilde{d}_2+\sqrt{\tilde{d}_2^2-4\tilde{d}_1\tilde{d}_2}}{2\tilde{d}_1},f_{\max}\} \\ {\min}\{\frac{-\tilde{d}_2+\sqrt{\tilde{d}_2^2-4\tilde{d}_1\tilde{d}_2}}{2\tilde{d}_1},f_{\max}\}, &\!\!\!\! \text { otherwise }\end{cases}\!.
\end{equation}

\section{Simulation Results}

In this section, we present numerical results to demonstrate the effectiveness of the proposed SADMM-SCA-based AO algorithm for the considered FDA-RIS-aided ISAC system. Under a three-dimensional Cartesian coordinate system, where the BS, the RIS and the target are located at (0,0,5)m, (0,10,5)m and (5,15,2)m, respectively. We consider that $K\!=\!4$ users are randomly distributed in a circle with the center (70,50,0)m and radius of $r_{\rm U}\!=\!12.5$m, and $C\!=\!3$ clutters are randomly distributed in a circle with the target as the center and radius of $r_{\rm c}\!=\!15$m. The numbers of BS transmit and receive antennas are set as $N_{\rm t}\!=\!8$ and $N_{\rm r}\!=\!4$, respectively, and the number of RIS reflecting elements is set as $M\!=\!64$. The large-scale path loss model for all channels is assumed to be $\beta_{\chi }\!=\!\beta_{0}(\frac{D_{\chi}}{D_0})^{-\alpha_{\chi}},\chi\!\in\!\{{\rm RU}\!,\!k,\,{\rm BR}\}$, where $\beta_0\!=\!-30$dB denotes the signal power attenuation at the reference distance $D_0\!=\!1$m and $\alpha_{\chi}$ denotes the path loss exponent. In terms of the BS-RIS link and each RIS-user $k$ link, we set $a_{\rm BR}\!=\!a_{{\rm RU},k}\!=\!2.3$. The small-scale Rician fading  is considered for all RIS-user $k$ channels with the Rician factor $\kappa\!=\!4$. In  addition,  for  the radar subsystem,  the overall path gains  associated with the target and clutter $c$ are  chosen as  $\beta_{\rm T}\!=\!\beta_{c}\!=\!10^{-6}$. Unless otherwise stated, the other system parameters are summarized in Table~\ref{table:by:fig}.

In order to demonstrate the effectiveness of our proposed SADMM-SCA-based AO algorithm (referred to as \textbf{Prop-AO}), we consider the following baselines: 1) \textbf{Comm-Centric}: The optimization variables $\{{\bf w}_k,{\bm\Theta},{f_{n_{\rm t}}}\}$ are jointly optimized to maximize the sum rate without considering the radar sensing performance; 2) \textbf{Radar-Centric}: The optimization variables $\{{\bf w}_k,{\bm\Theta},{f_{n_{\rm t}}},{\bf u}\}$ are jointly optimized to maximize SCNR without considering the communication performance; 3) \textbf{PA}: In this scheme, the BS is equipped with the conventional phased array of the same size as the FDA in \textbf{Prop-AO}; 4) \textbf{SDR-RIS}: The optimal RIS phase shift matrix ${\bm\Theta}$ is obtained by jointly leveraging the SDR and Gaussian Randomization methods\cite{sim_RIS_SDR}; 5) \textbf{Fix-FDA}: In this scheme, the fixed uniform frequency offset, i.e., $f_{n_{\rm t}}\!=\!f_{\rm ref}\!+\!(n_{\rm t}\!-\!1)\Delta f$, is adopted for the FDA\cite{Samefre}.

\begin{table}[t]
	\centering
	\vspace{-12mm}
	\caption{Basic system parameters}\label{table:by:fig} 
	\vspace{-1.5mm}
  \resizebox{0.56\textwidth}{!}{
  \begin{tabular}{|c|c|} \hline 
  Parameter & Value \\ \hline
  BS transmit power $P_{\rm B}$ & 35dBm \\ \hline   	 				
   Thermal noise power at users/ BS $\sigma_{k}^2/\sigma_{\rm R}^2$ & -80dBm/ -70dBm  \\ \hline
    Reference carrier frequency $f_{\rm ref}$ & $10$GHz \\ \hline
    Maximum frequency offset $f_{\rm max}$ & $8$MHz \\ \hline
    SCNR threshold $\gamma_{\rm T}$ & 20dB \\ \hline
  \end{tabular} }
  \vspace{-6mm}
\end{table}

\begin{figure*}[t]\vspace{-8mm}
	%\addtocounter{figure}{+1}
	\centering
	\subfigcapskip=-10pt
	\subfigure[(a)]
	{ \includegraphics[width=0.46\textwidth]{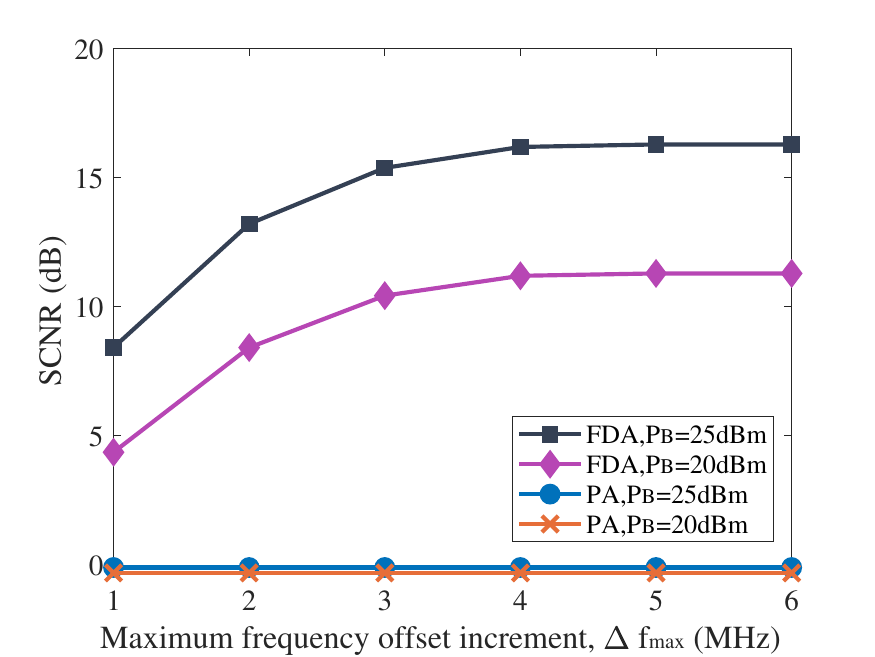}
    \label{Sim_SUST}
	}
	\subfigure[(b)]{
		\includegraphics[width=0.46\textwidth]{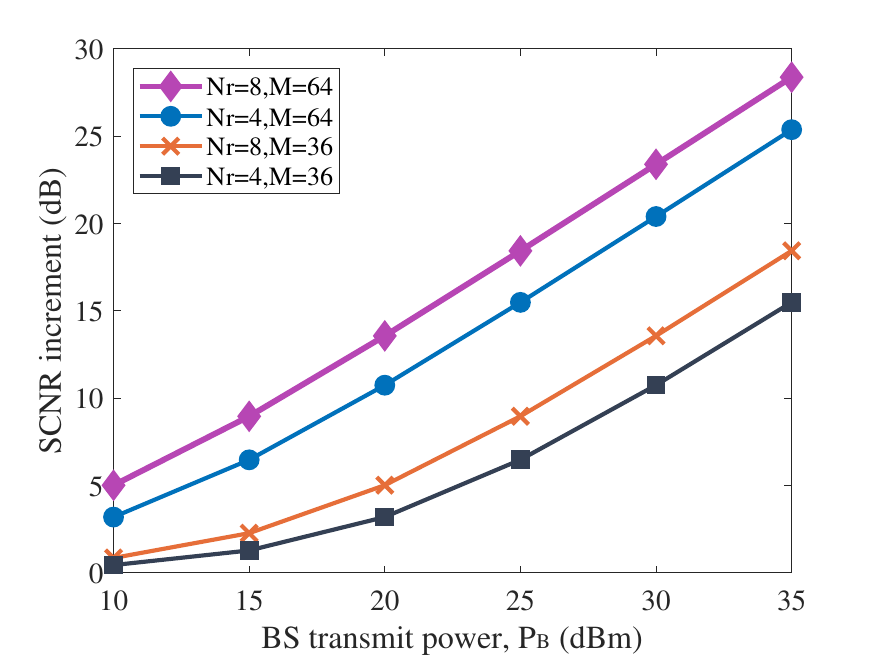}
        \label{Sim_deltaSUST}
	}
	\vspace{-2mm}
	\caption{(a) SCNR of the SUST system versus the maximum allowable frequency offset increment $\Delta f_{\max}$. (b) SCNR increment versus the BS transmit power $P_{\rm B}$ with $\Delta f_{\max}\!=\!3$MHz.}
	\vspace{0mm}
\end{figure*}

For the SUST system, Fig.~\ref{Sim_SUST} shows the achievable SCNRs of the FDA-based and PA-based schemes versus the maximum allowable frequency offset increment $\Delta f_{\max}$ under different BS transmit power budgets, i.e., $P_{\rm B}\!=\!25$dBm and $P_{\rm B}\!=\!20$dBm, where the distance between the target and the clutter is chosen as $\Delta D\!=\!4$m. It is clearly observed that for each $P_{\rm B}$, the achievable SCNR in the FDA-based scheme is significantly higher than that in the PA-based scheme. Furthermore, the SCNR achieved by the FDA-based scheme monotonically increases with $\Delta f_{\max}$ in the low $\Delta f_{\max}$ region, and then becomes saturated in the  high $\Delta f_{\max}$ region, which is well consistent with the derived optimal frequency offset increment in Proposition~\ref{prop-optwu}. For the PA, the obtained SCNR is independent of $\Delta f_{\max}$. Fig.~\ref{Sim_deltaSUST} further plots the SCNR increment of FDA versus the BS transmit power $P_{\rm B}$ under different numbers of BS receive antennas $N_{\rm r}$ and RIS reflecting elements $M$. It is shown that the SCNR increment under each pair $(N_{\rm r},M)$ monotonically increases with $P_{\rm B}$. Moreover, increasing $N_{\rm r}$ also enables the SCNR increment, which is consistent with the obtained insights in Section~\ref{FDAPA}-B.  These results in Fig.~\ref{Sim_SUST} and Fig.~\ref{Sim_deltaSUST} demonstrate the advantage of FDA in suppressing clutter echos as compared to the conventional PA, which can be attributed to its additional spatial resolution in the distance domain.

\begin{figure*}[t]\vspace{-5.5mm}
	%\addtocounter{figure}{+1}
	\centering
	\subfigcapskip=-10pt
	\subfigure[(a)]
	{ \includegraphics[width=0.46\textwidth]{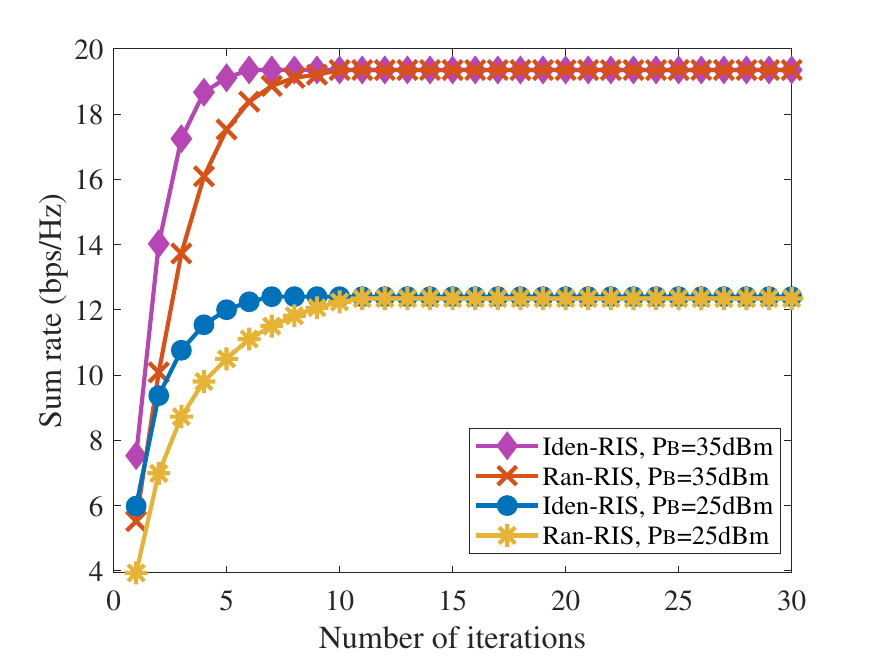}
    \label{Sim_Conver}
	}
	\subfigure[(b)]{
		\includegraphics[width=0.46\textwidth]{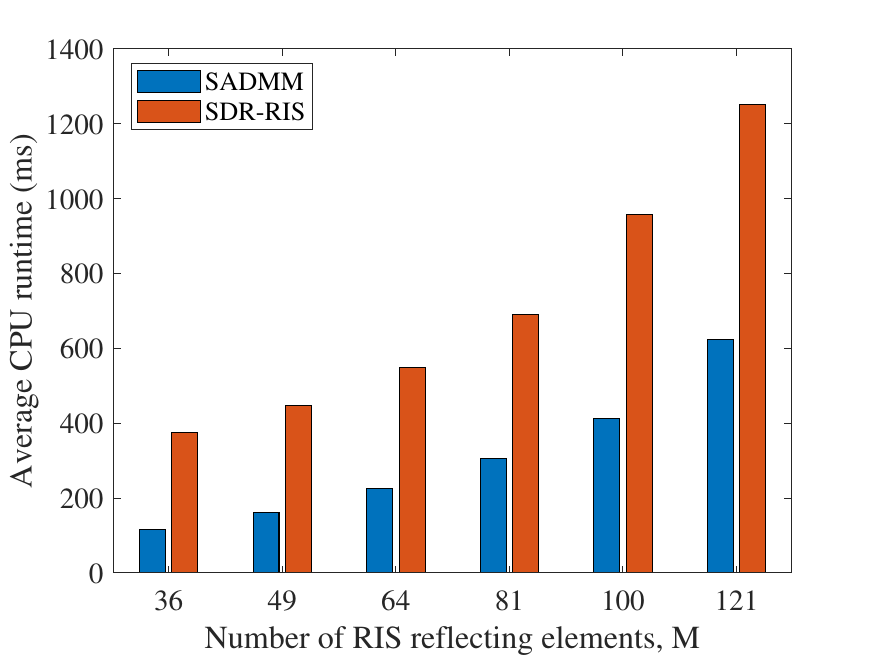}
        \label{Sim_timing}
	}
	\vspace{-2mm}
	\caption{(a) Convergence behaviors of respective algorithms. (b) Average CPU runtime versus the number of RIS reflecting elements $M$.}
	\vspace{-7mm}
\end{figure*}

Fig.~\ref{Sim_Conver} shows convergence behaviors of \textbf{Prop-AO} under different BS transmit power budgets, i.e., $P_{\rm B}\!=\!25$dBm and $P_{\rm B}\!=\!35$dBm, where two different initialization schemes, i.e., \textbf{Iden-RIS} and \textbf{Ran-RIS}, are considered with the RIS phase shift matrix ${\bm\Theta}$  initialized as an identity matrix and a random matrix, respectively. It is shown that for each $P_{\rm B}$, although \textbf{Ran-RIS} converges slower than \textbf{Iden-RIS}, both the two schemes are able to converge to the same sum rate within 15 iterations.  Moreover, Fig.~\ref{Sim_timing} shows the average CPU runtime of  the proposed SADMM algorithm and  \textbf{SDR-RIS}  for solving problem (P3-3-2) versus the number of RIS reflecting elements $M$. It can be clearly observed that although the time consumption for both schemes increases with $M$, the time consumption for the proposed SADMM algorithm is noticeably lower than that of \textbf{SDR-RIS}, thereby demonstrating its high efficiency.

\iffalse
\vspace{0mm}
\begin{figure}[t]
	%\addtocounter{figure}{+1}
	\centering
	\subfigcapskip=-10pt
	{ \includegraphics[width=0.46\textwidth]{Fig/PB_v1.pdf}
	}
	\caption{Sum rate  versus the BS transmit power $P_{\rm B}$.}\label{Sim_PB1}
	\vspace{-6mm}
\end{figure}

\begin{figure}[t]
	%\addtocounter{figure}{+1}
	\centering
	\subfigcapskip=-10pt
	{ \includegraphics[width=0.46\textwidth]{Fig/tradeoff_v1.pdf}
	}
	\caption{Sum rate  versus the SCNR threshold $\gamma_{\rm T}$.}\label{Sim_tradeoff1}
	\vspace{-8mm}
\end{figure}
\fi

\begin{figure} 
    \vspace{-0mm}
    \begin{minipage}[t]{0.5\linewidth} 
    \centering 
    \includegraphics[width=0.9\textwidth]{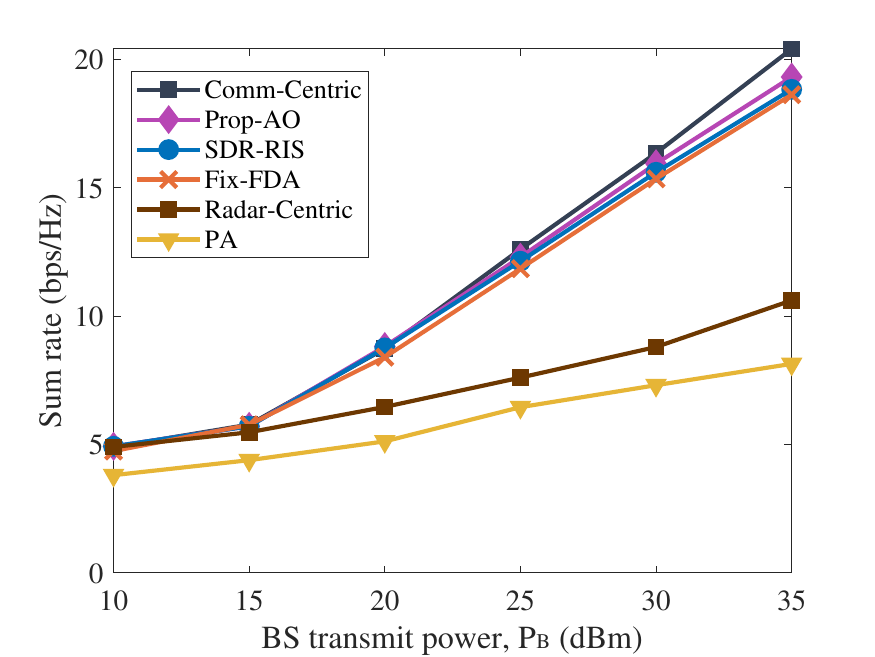} 
    \vspace{-5mm}
    \caption{Sum rate  versus the BS transmit power $P_{\rm B}$.} 
    \label{Sim_PB} 
    \end{minipage}% 
    \begin{minipage}[t]{0.5\linewidth} 
    \centering 
    \includegraphics[width=0.9\textwidth]{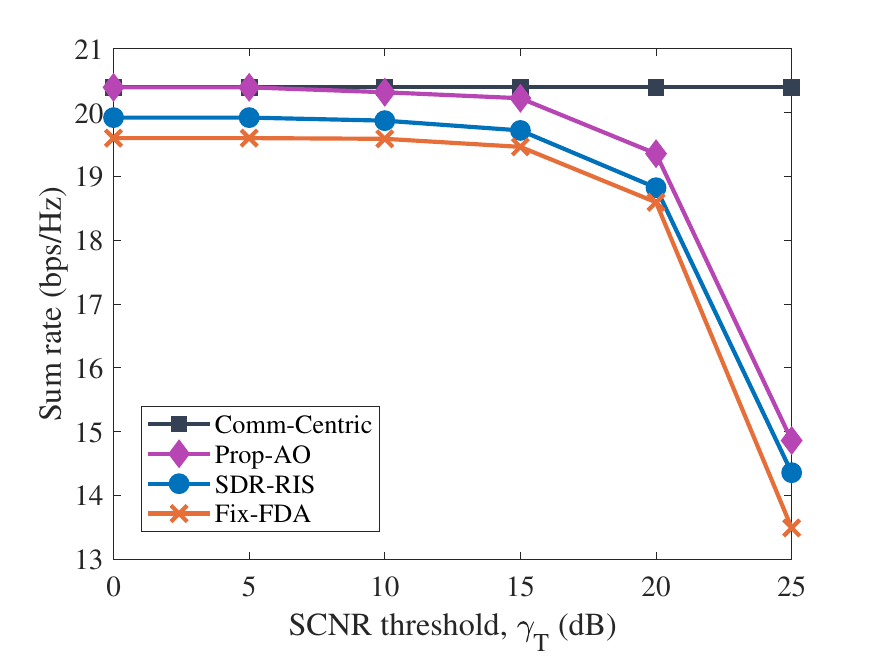} 
    \vspace{-5mm}
    \caption{Sum rate  versus the SCNR threshold $\gamma_{\rm T}$.} 
    \label{Sim_tradeoff} 
    \end{minipage} 
    \vspace{-3mm}
    \end{figure}
    
    \begin{figure}[t]
        %\addtocounter{figure}{+1}
        \vspace{0mm}
        \centering
        \subfigcapskip=-10pt
        { \includegraphics[width=0.46\textwidth]{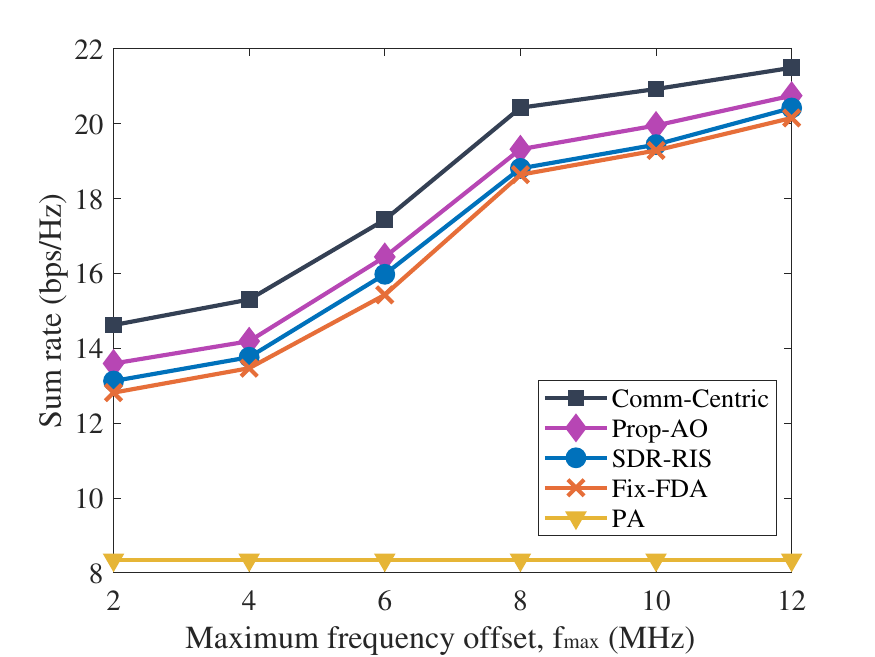}
        }
        \vspace{-5mm}
        \caption{Sum rate  versus the maximum allowable frequency offset $f_{\max}$.}\label{Sim_fre}
        \vspace{-8mm}
    \end{figure}

In Fig.~\ref{Sim_PB}, we show the achievable sum rates of all studied schemes versus the BS transmit power $P_{\rm B}$. For each scheme, the achievable sum rate monotonically increases with $P_{\rm B}$. \textbf{Prop-AO} achieves the best communication performance among all ISAC schemes but naturally performs worse than \textbf{Comm-Centric}, since it is not designed dedicatedly for enhancing communication performance. Furthermore, the performance loss of \textbf{Fix-FDA} can be attributed to the limited DoFs in the frequency offsets optimization, and the inferior performance of \textbf{Radar-Centric} is due to the fact that it is dedicated to improving the target detection  performance. Moreover, we observe that the performance of \textbf{PA} is considerably worse than the others, since the PA provides only angle-dependent beampattern, which cannot effectively suppress clutter interference, especially when the target and the clutters are in close spatial proximity.

Fig.~\ref{Sim_tradeoff} illustrates the achievable sum rates of respective schemes versus the SCNR threshold $\gamma_{\rm T}$. Since \textbf{Comm-Centric} is dedicated to enhancing communication performance, it is  not influenced by the growth of $\gamma_{\rm T}$ and thus serves as the upper bound. The achievable sum rate of \textbf{Prop-AO} is the same as that of \textbf{Comm-Centric} in the low SCNR region, whereas gradually decreases with the increase of $\gamma_{\rm T}$. This is because the SCNR constraint in problem (P3) will become tight for a sufficiently large $\gamma_{\rm T}$. In this context, with the rise of $\gamma_{\rm T}$, the target detection performance is guaranteed at the expense of a certain communication performance loss.

In Fig.~\ref{Sim_fre}, we further examine the effects of the maximum allowable frequency offset $f_{\rm max}$ on the sum rate performance. It is clear that for the FDA-based schemes, i.e., \textbf{Comm-Centric}, \textbf{Prop-AO}, \textbf{SDR-RIS} and \textbf{Fix-FDA}, the sum rate monotonically grows with $f_{\rm max}$ owing to the increased DoFs in frequency offsets design. For the conventional \textbf{PA}, the achievable sum rate is naturally irrelevant to $f_{\rm max}$ and also noticeably lower than those achieved by the FDA-based schemes, since the PA's spatial resolution is limited to the angular domain.

\vspace{-3mm}
\section{Conclusions}

In this paper, we have investigated an FDA-RIS-aided ISAC system. Specifically, we have maximized sum rate by jointly optimizing the BS transmit beamforming vectors, the covariance matrix of the dedicated radar signal, the RIS phase shift matrix, the FDA frequency offsets and the radar receive equalizer, while guaranteeing the level of SCNR of the echo signal. To tackle this intractable non-convex problem, we have first theoretically proved that the dedicated radar signal is unnecessary. Based on this fact, we turned our attention to a simple SUST scenario to theoretically demonstrate that  the FDA-aided BS can always achieve a higher SCNR than its PA-aided counterpart due to its additional spatial resolution in the distance domain. Moreover, it has been revealed that the SCNR increment exhibits linear growth with the BS transmit power and the number of BS receive antennas. Then, we have reformulated the simplified problem into a tractable parametric subtraction form  using the FP-based technique, and subsequently proposed an efficient SADMM-SCA-based AO algorithm to  find the locally optimal solution. Finally, numerical results illustrated superior communication and sensing performance of our proposed algorithm.

\vspace{-4mm}
\appendix
\section*{ }
\vspace{-8.1mm}
\subsection{Proof of Proposition~\ref{prop-sim}}
\label{app-prop-sim}
We prove proposition \ref{prop-sim} by contradiction. First, we assume that the optimal solution to problem (P1) is ${\mathcal{V}_{\rm P_1}^{\rm opt}}\!\!=\!\!\{{\bf w}_{ k}^{\rm opt},{\bf R}_{0}^{\rm opt},{{\bm \Theta}^{\rm opt}},\Delta f^{\rm opt}_{ n_{\rm t}},\!{\bf u}^{\rm opt}\}$ with ${\bf R}_{0}^{\rm opt}\!\neq\!{\bf 0}$, under which the SINR of the user $k$, the BS transmit power, and the SCNR are ${\gamma}_k^{\rm opt}$, ${P}_{\rm B}^{\rm opt}$ and ${\Gamma}_{\rm T}^{\rm opt}$, respectively. Hereafter, we aim to show that there always exists another solution ${\check{\mathcal{V}}_{\rm P_1}}\!=\!\{\check{\bf w}_{ k}, \check{\bf R}_{0}, {{\bm \Theta}^{\rm opt}},\Delta {f}^{\rm opt}_{ n_{\rm t}},{\bf u}^{\rm opt}\}$ with $\check{\bf R}_{0}\!=\!{\bf 0}$, which also leads to the same (or a higher) SINR for each user, the same BS transmit power and SCNR. Specifically, we first formulate the following SDR problem with the  dedicated radar sensing covariance matrix  ${\bf R}_0$ neglected.
\begin{subequations}\label{P1-1-SDP}
    \begin{align}
        ({\text{P1-SDR}}):~  
        {{\rm find}}~~&\{\check{\bf W}_{ k}\} \\
        {\rm {s.t.}}~~&{\rm Tr}(\tilde{\bf h}_k\tilde{\bf h}_k^{\rm H}\check{\bf W}_k)\!\geq\!\gamma_k^{\rm opt}\big(\sum\nolimits_{k'\neq k}{\rm Tr}(\tilde{\bf h}_k\tilde{\bf h}_k^{\rm H}\check{\bf W}_{k'})\!+\!\sigma_k^2\big),~\forall k\!\in\!\mathcal{K},\label{P1-1-SDP-1} \\
        &\sum\nolimits_{k=1}^{K}{\rm Tr}(\check{\bf W}_k)\!=\!P_{\rm B}^{\rm opt}, \\
        &\sum\nolimits_{k=1}^{K}{\rm Tr}({\bf G}_{\rm T}\check{\bf W}_k)\!=\!{\Gamma}_{\rm T}^{\rm opt}\big(\sum\nolimits_{k=1}^{K}\sum\nolimits_{c=1}{\rm Tr}({\bf G}_{c}\check{\bf W}_k)\!+\!\Vert{\bf u}^{\rm opt}\Vert^2\sigma_{\rm R}^2\big),\label{P1-1-SDP-2}\\
        &\check{\bf W}_k\succeq {\bf 0},~\forall k\!\in\!\mathcal{K},
    \end{align}
\end{subequations}
where ${\bf G}_{\rm T}\!=\!\beta_{\rm T}{\bf H}_{\rm BT}^{\rm H}{\bf u}^{\rm opt}({\bf u}^{\rm opt})^{\rm H}{\bf H}_{\rm BT}$ and ${\bf G}_{c}\!=\!\beta_{c}{\bf H}_{{\rm BC},c}^{\rm H}{\bf u}^{\rm opt}({\bf u}^{\rm opt})^{\rm H}{\bf H}_{{\rm BC},c}$. Note that problem (P1-SDR) is feasible, since we can define a feasible solution as $\check{\bf W}_{k_0}\!=\!{\bf w}_{k_0}^{\rm opt}({\bf w}_{k_0}^{\rm opt})^{\rm H}\!+\!{\bf R}_{0}^{\rm opt}$ and $\check{\bf W}_{k'}\!=\!{\bf w}_{k'}^{\rm opt}({\bf w}_{k'}^{\rm opt})^{\rm H},{k'}\!\neq\!{k_0}$ for any $k_0$, 
which yields a higher SINR for user $k_0$ while holding SINRs of  other users fixed, i.e., $\check{\gamma}_{k_0}\!>\!\gamma_{k_0}^{\rm opt}$ and $\check{\gamma}_k\!=\!\gamma_{k}^{\rm opt},\forall k\!\neq\!k_0$ with $\check{\gamma}_k$ representing the SINR for user $k$.  Motivated by  this fact, a higher  sum rate  
$\check{R}_{\rm sum}$ than that achieved by solutions in $\mathcal{V}_{\rm P_{1}}^{\rm opt}$ can be attained. Furthermore, according to \cite[Theorem 3.2]{rank-SDR}, there always exists an optimal solution to problem (P1-SDR) that satisfies $\sum\limits_{k=1}^{K}({\rm Rank}(\check{\bf W}_k^{\rm opt}))^2\!\leq\!K\!+\!2$. Meanwhile, since $\gamma_k^{\rm opt}\!>\!0$, there must be $\check{\bf W}_k^{\rm opt}\!\neq\!{\bf 0}$ or equivalently ${\rm Rank}(\check{\bf W}_k^{\rm opt})\!\geq\!1$. Therefore, ${\rm Rank}(\check{\bf W}_k^{\rm opt})\!=\!1,\forall k\!\in\!\mathcal{K}$ must exist for problem (P1-SDR) and the rank-1 $\check{\bf W}_{k}^{\rm opt}$ can be found by referring to \cite{rank-SDR}. By applying eigenvalue decomposition to $\check{\bf W}_k^{\rm opt}$, the optimal BS beamforming vectors $\check{\bf w}_k^{\rm opt}$ can be obtained. It follows from \eqref{P1-1-SDP-1}-\eqref{P1-1-SDP-2} that  the optimal solutions in ${\check{\mathcal{V}}_{\rm P_1}}$ are able to lead to the same (or a higher) objective value of the original problem (P1) as those in ${\mathcal{V}_{\rm P_1}^{\rm opt}}$, with constraints \eqref{P1cons1}-\eqref{P1cons4} also guaranteed,  which contradicts the assumption that ${\mathcal{V}_{\rm P_1}^{\rm opt}}$ is the optimal solution to problem  (P1). This completes the proof.

\vspace{-5mm}
\subsection{Proof of Proposition~\ref{prop-optwu}}
\label{app-optwu}

Firstly, for any given $\{{\bf w}_{\rm c},{\Delta f}\}$, problem (P2) w.r.t. ${\bf u}$ is a generalized Rayleigh quotient problem, to which the optimal ${\bf u}^{\rm opt}$ is obtained as ${\bf u}^{\rm opt}\!=\!\beta_{{\bf u}}({\beta_{\rm C}}\vert{\bf b}_{\rm BC}^{\rm H}{\bf w}_{\rm c}\vert^2\vert p_{\rm clu}\vert^2{\bf b}_{\rm RB,r}{\bf b}_{\rm RB,r}^{\rm H}\!+\!\sigma_{\rm R}^2{\bf I}_{N_{\rm r}})^{-1}{\bf b}_{\rm RB,r}$\cite{Rayquo}, where $\beta_{{\bf u}}$ is the normalization factor leading to $\Vert{\bf u}^{\rm opt}\Vert^2\!=\!1$ and ${p}_{\rm clu}\!=\!\beta_{\rm BR}{\bf b}_{\rm BR}^{\rm H}{\bm\Theta}{\bf a}_{\rm clu,0}^{\rm r}({\bf a}_{\rm clu,0}^{\rm t})^{\rm H}{{\bm\Theta}}{\bf b}_{\rm BR}$.

By substituting ${\bf u}$ into $\Gamma_{\rm SUST}$, we have
\begin{align}\label{hat-SUST}
    \hat{\Gamma}_{\rm SUST}\!&=\!{\beta}_{\rm T}\vert{\bf b}_{\rm BT}^{\rm H}{\bf w}_{\rm c}\vert^2\vert p_{\rm tar}\vert^2{\bf b}_{\rm RB,r}^{\rm H}({\beta_{\rm C}}\vert{\bf b}_{\rm BC}^{\rm H}{\bf w}_{\rm c}\vert^2\vert p_{\rm clu}\vert^2{\bf b}_{\rm RB,r}{\bf b}_{\rm RB,r}^{\rm H}\!+\!\sigma_{\rm R}^2{\bf I}_{N_{\rm t}})^{-1}{\bf b}_{\rm RB,r}\notag\\
    &\overset{(a)}{=}\!\frac{{\beta}_{\rm T}N_{\rm r}\vert{p}_{\rm tar}\vert^2\vert{\bf b}_{\rm BT}^{\rm H}{\bf w}_{\rm c}\vert^2}{{\beta_{\rm C}N_{\rm r}\vert p_{\rm tar}\vert^2\vert{\bf b}_{\rm BC}^{\rm H}{\bf w}_{\rm c}\vert^2}\!+\!\sigma_{\rm R}^2},
\end{align}
where ${p}_{\rm tar}\!=\!\beta_{\rm BR}{\bf b}_{\rm BR}^{\rm H}{\bm\Theta}{\bf a}_{\rm tar,0}^{\rm r}({\bf a}_{\rm tar,0}^{\rm t})^{\rm H}{{\bm\Theta}}{\bf b}_{\rm BR}$, $(a)$ holds due to the equality $({\bf A}\!+\!{\bf x}{\bf y}^{\rm H})^{-1}\!=\!{\bf A}^{-1}\!-\!\frac{{\bf A}^{-1}{{\bf x}{\bf y}^{\rm H}}{\bf A}^{-1}}{(1+{{\bf y}^{\rm H}{\bf A}^{-1}{\bf x}})}$ for an invertible matrix ${\bf A}$. The equalities  $\Vert{\bf b}_{\rm RB,r}\Vert^2\!=\!N_{\rm r}$ and $\vert{p}_{\rm clu}\vert^2\!=\!\vert{p}_{\rm tar}\vert^2$ are also applied to $(a)$. Based on the above discussion, problem (P2) reduces to the following problem w.r.t. $\{{\bf w}_{\rm c},\Delta f\}$
\begin{subequations}
\begin{align}
    ({\text{P2-1}}):~  \underset{{\bf w}_{\rm c},\Delta f}
    {{\max}}~&\tilde{\Gamma}_{\rm SUST}\!\triangleq\!\frac{{\beta}_{\rm T}N_{\rm r}\vert p_{\rm tar}\vert^2\vert{\bf b}_{\rm BT}^{\rm H}{\bf w}_{\rm c}\vert^2}{{\beta_{\rm C}N_{\rm r}\vert p_{\rm tar}\vert^2\vert{\bf b}_{\rm BC}^{\rm H}{\bf w}_{\rm c}\vert^2}\!+\!\frac{\sigma_{\rm R}^2}{P_{\rm B}}\Vert{\bf w}_{\rm c}\Vert^2},\\
    {\rm {s.t.}}~&
    \Vert {\bf w}_{\rm c}\Vert^2\!\leq\!P_{\rm B},~\Delta f\in(0,\Delta f_{\max}].
\end{align}
\end{subequations}
    It is readily verified that the BS transmit power constraint must be active at the optimum, i.e., $\Vert {\bf w}_{\rm c}^{\rm opt}\Vert^2\!=\!P_{\rm B}$, and  thus problem  (P2-1) with any given $\Delta f$ can also be regarded as a generalized Rayleigh quotient problem w.r.t. ${\bf w}_{\rm c}$, to which the optimal ${\bf w}_{\rm c}^{\rm opt}$ is obtained as ${\bf w}_{\rm c}^{\rm opt}\!=\!\beta_{{\bf w}_{\rm c}}({\beta}_{\rm C}N_{\rm r}\vert p_{\rm tar}\vert^2{\bf b}_{\rm BC}{\bf b}_{\rm BC}^{\rm H}\!+\!\frac{\sigma_{\rm R}^2}{P_{\rm B}}{\bf I}_{N_{\rm t}})^{-1}{\bf b}_{\rm BT}$, where $\beta_{{\bf w}_{\rm c}}$ is the normalization factor satisfying $\Vert{\bf w}_{\rm c}^{\rm opt}\Vert^2\!=\!P_{\rm B}$. Then, substituting ${\bf w}_{\rm c}^{\rm opt}$ into $\tilde{\Gamma}_{\rm SUST}$ yields
\begin{align}\label{prop_SUST_FDA}
    {\Gamma}_{\rm SUST}^{\rm FDA}\!&=\!{\beta}_{\rm T}N_{\rm r}\vert p_{\rm tar}\vert^2 {\bf b}_{\rm BT}^{\rm H}({\beta_{\rm C}N_{\rm r}\vert p_{\rm tar}\vert^2{\bf b}_{\rm BC}{\bf b}_{\rm BC}^{\rm H}}\!+\!{\sigma_{\rm R}^2}/{P_{\rm B}}{\bf I}_{N_{\rm t}})^{-1}{\bf b}_{\rm BT}\notag\\
    &\overset{(b)}{=}\!{{\beta}_{\rm T}N_{\rm r}\vert p_{\rm tar}\vert^2P_{\rm B}}\big(\frac{N_{\rm t}}{\sigma_{\rm R}^2}\!-\!\frac{P_{\rm B}\beta_{\rm C}N_{\rm r}\vert p_{\rm tar}\vert^2\big\vert\frac{\sin({N_{\rm t}}{2\pi\Delta f\Delta D}/{{\rm c}})}{\sin({2\pi\Delta f\Delta D}/{{\rm c}})}\big\vert^2}{{\sigma_{\rm R}^2}(\sigma_{\rm R}^2\!+\!P_{\rm B}N_{\rm t}\beta_{\rm C}N_{\rm r}\vert p_{\rm tar}\vert^2)}\big),
\end{align}
where $(b)$ holds similarly to $(a)$ in \eqref{hat-SUST}. Also, the equalities $\Vert{\bf b}_{\rm BT}\Vert^2\!=\!\Vert{\bf b}_{\rm BC}\Vert^2\!=\!N_{\rm t}$ and $\vert{\bf b}_{\rm BT}^{\rm H}{\bf b}_{\rm BC}\vert^2\!=\!\big\vert\frac{\sin({N_{\rm t}}{2\pi\Delta f\Delta D}/{{\rm c}})}{\sin({2\pi\Delta f\Delta D}/{{\rm c}})}\big\vert^2$ with $\Delta D\!\triangleq\!\vert D_{\rm RC}\!-\!D_{\rm RT}\vert$ are applied to $(b)$.

Based on \eqref{prop_SUST_FDA}, problem (P2) finally reduces to the following optimization problem w.r.t. $\Delta f$ 
    \begin{align}
        ({\text{P2-2}}):~  \underset{\Delta f}
        {{\max}}~{\Gamma}_{\rm SUST}^{\rm FDA},~~
        {\rm {s.t.}}~\Delta f\!\in\!(0,\Delta f_{\max}].
    \end{align}
By leveraging the fact that the function $f_{\rm SUST}(\Delta f)\!\triangleq\!\Big\vert\frac{\sin({N_{\rm t}}{2\pi\Delta f\Delta D}/{{\rm c}})}{\sin({2\pi\Delta f\Delta D}/{{\rm c}})}\Big\vert^2$ monotonically decreases with $\Delta f$ when $0\!<\!\Delta f\!\leq\!\Delta f_0$, where $\Delta f_0\!=\!\frac{\rm c}{2N_{\rm t}\Delta D}$ is the smallest positive zero of $f_{\rm SUST}(\Delta f)$, we readily derive the optimal solution to problem (P2-2) as
$\Delta f^{\rm opt}\!=\!\min\{\Delta f_{\max},\Delta f_0\}$. Furthermore, substituting $\Delta f^{\rm opt}$ into $\Gamma_{\rm SUST}^{\rm FDA}$ yields $\Gamma_{\rm SUST,max}^{\rm FDA}$.

In particular, for its PA-RIS-aided counterpart, by substituting $\Delta f\!\!=\!0$ into ${\Gamma}_{\rm SUST}^{\rm FDA}$, we have 
    ${\Gamma}_{\rm SUST}^{\rm PA}\!=\!\frac{{\beta}_{\rm T}P_{\rm B}N_{\rm r}\vert p_{\rm tar}\vert^2}{\sigma_{\rm R}^2}\big(N_{\rm t}\!-\!\frac{P_{\rm B}\beta_{\rm C}N_{\rm t}^2N_{\rm r}\vert p_{\rm tar}\vert^2}{\sigma_{\rm R}^2\!+\!P_{\rm B}N_{\rm t}\beta_{\rm C}N_{\rm r}\vert p_{\rm tar}\vert^2}\big)$.
It can be verified that $\big\vert\frac{\sin({N_{\rm t}}{2\pi\Delta f\Delta D}/{{\rm c}})}{\sin({2\pi\Delta f\Delta D}/{{\rm c}})}\big\vert^2\!\leq\!N_{\rm t}^2$ always holds for any $\Delta f\!>\!0$, thus leading to ${\Gamma}_{\rm SUST}^{\rm PA}\!\leq\!{\Gamma}_{\rm SUST}^{\rm FDA}$, and the maximum achievable SCNR increment is given by
    $\Delta\Gamma_{\max}\!=\!{\Gamma}_{\rm SUST,max}^{\rm FDA}\!-\!{\Gamma}_{\rm SUST}^{\rm PA}\!=\!\frac{{\beta}_{\rm T}{\beta}_{\rm C}N_{\rm r}^2\vert p_{\rm tar}\vert^4P_{\rm B}^2(N_{\rm t}^2-\big\vert\frac{\sin({N_{\rm t}}{2\pi\Delta f^{\rm opt}\Delta D}/{{\rm c}})}{\sin({2\pi\Delta f^{\rm opt}\Delta D}/{{\rm c}})}\big\vert^2)}{\sigma_{\rm R}^2(\sigma_{\rm R}^2+P_{\rm B}N_{\rm t}\beta_{\rm C}N_{\rm r}\vert p_{\rm tar}\vert^2)}$.
This completes the proof.

\vspace{-4mm}
\subsection{Proof of Proposition~\ref{prop1}}
\label{appA}
Firstly, since the auxiliary variable $\alpha_k$ in problem (P3) is unconstrained, its optimal solution for any given $\{{\bf w}_{k},{\bm\Theta},w_{k},{\bf u},f_{n_{\rm t}}\}$ can be directly obtained via the first-order optimality condition as $\alpha_k^{\rm opt}\!=\!\frac{\sqrt{1+w_k}\tilde{\bf h}_k^{\rm H}{\bf w}_k}{\sum\nolimits_{k'=1}^{K}\vert\tilde{\bf h}_{k}^{\rm H}{\bf w}_{k'}\vert^2\!+\!\sigma_k^2}$. By substituting $\{\alpha_{ k}^{\rm opt}\}_{k\in\mathcal{K}}$ into \eqref{P2obj}, we have $f_{\rm FP}(\{w_k\})\!\triangleq\!\sum\nolimits_{k=1}^{K}\big(\log(1\!+\!w_k)\!-\!w_k\!+\!(1\!+\!w_k)\frac{\vert\tilde{\bf h}_k^{\rm H}{\bf w}_k\vert^2}{\sum\nolimits_{k'=1}^{K}\vert\tilde{\bf h}_{k}^{\rm H}{\bf w}_{k'}\vert^2\!+\!\sigma_k^2}\big)$. Moreover, as the variable $w_k$ in problem (P3) is also unconstrained, its optimum can be similarly obtained as $w_k^{\rm opt}\!=\!\frac{\vert\tilde{\bf h}_k^{\rm H}{\bf w}_k\vert^2}{\sum\nolimits_{k'\neq k}\vert\tilde{\bf h}_{k}^{\rm H}{\bf w}_{k'}\vert^2\!+\!\sigma_k^2}$. Substituting $\{w_k^{\rm opt}\}_{k\in\mathcal{K}}$ into $f_{\rm FP}(\{w_{k}\})$ further yields $f_{\rm FP}(\{w_k^{\rm opt}\})\!=\!\sum\nolimits_{k=1}^{K} R_k$. Therefore, we can conclude that  problems (P1-1)  and  (P3) are equivalent via  introducing  the  auxiliary variables $\{\alpha_{k},w_{k}\}$. This completes the proof.

\vspace{-4mm}
\subsection{Proof of Proposition~\ref{prop-new-new-2}}
\label{app-prop-new-new-2}
We prove this proposition by contradiction. Specifically, we assume the optimal solutions to problem (P3) to be $\bar{\mathcal{V}}_{\rm P_3}\!=\!\{\bar{\alpha}_{ k},\bar{w}_{ k},{\bar{\bf w}}_{ k},\bar{\bm \Theta},\Delta{ \bar{f}}_{ n_{\rm t}},\bar{\bf u}\}$, which leads to $\sum\limits_{k=1}^{K}\Vert\bar{\bf w}_k\Vert^2\!<\!P_{\rm B}$. Then we can always find another set of solutions $\hat{\mathcal{V}}_{\rm P_3}\!=\!\{\hat{\alpha}_{ k},\bar{w}_{ k},{\hat{\bf w}}_{ k},\bar{\bm \Theta},\bar\Delta{ f}_{ n_{\rm t}},\bar{\bf u}\}$ with $\hat{\alpha}_k\!=\!\frac{\bar{\alpha}}{\sqrt{\nu_1} }$, $\hat{\bf w}_k\!=\!{\sqrt{\nu_1} }\bar{\bf w}_k$ and $\nu_1\!=\!\frac{P_{\rm B}}{\sum\limits_{k=1}^{K}\Vert\bar{\bf w}_k\Vert^2}\!>\!1$, which leads to 
\begin{align}\label{contra1}
    \hat{f}_{\rm obj,P_3}\!&=\!\sum\limits_{k=1}^{K}\big( \log(1\!+\!\bar{w}_k)\!-\!\bar{w}_k\!+\!2\sqrt{1\!+\!\bar{w}_k}\Re\{\bar{\alpha}_k^*\tilde{\bf h}_k^{\rm H}\bar{\bf w}_k\}\!-\!\vert\bar{\alpha}_k\vert^2\big(\sum\limits_{k'=1}^{K}\vert\tilde{\bf h}^{\rm H}_k\bar{\bf w}_{k'}\vert^2\!+\!\frac{\sigma_k^2}{\nu_1}\big) \big)\notag\\
    &>\!\sum\limits_{k=1}^{K}\big( \log(1\!+\!\bar{w}_k)\!-\!\bar{w}_k\!+\!2\sqrt{1\!+\!\bar{w}_k}\Re\{\bar{\alpha}_k^*\tilde{\bf h}_k^{\rm H}\bar{\bf w}_k\}\!-\!\vert\bar{\alpha}_k\vert^2\big(\sum\limits_{k'=1}^{K}\vert\tilde{\bf h}^{\rm H}_k\bar{\bf w}_{k'}\vert^2\!+\!{\sigma_k^2}\big) \big)\!=\!\bar{f}_{\rm obj,P_3},\notag
\end{align}
\begin{align}
        \hat{\Gamma}_{\rm T}\!&=\!\frac{\nu_1\beta_{\rm T}\sum\limits_{k=1}^{K}\vert\bar{\bf u}^{\rm H}{\bf H}_{\rm BT}\bar{\bf w}_k\vert^2\!}{\nu_1\sum\limits_{c=1}^{C}\!\sum\limits_{k=1}^{K}\!\beta_{c}\vert\bar{\bf u}^{\rm H}{\bf H}_{{\rm BC},c}\bar{\bf w}_k\vert^2\!+\!\Vert\bar{\bf u}\Vert^2\sigma_{\rm R}^2}\!>\!\frac{{\beta}_{\rm T}\sum\limits_{k=1}^{K}\vert\bar{\bf u}^{\rm H}{\bf H}_{\rm BT}\bar{\bf w}_k\vert^2\!}{\sum\limits_{c=1}^{C}\sum\limits_{k=1}^{K}{\beta}_{c}\vert\bar{\bf u}^{\rm H}{\bf H}_{{\rm BC},c}\bar{\bf w}_k\vert^2\!+\!\Vert\bar{\bf u}\Vert^2\sigma_{\rm R}^2}\!=\!\bar{\Gamma}_{\rm T},
\end{align}
where $\{\hat{f}_{\rm obj,P_3},\hat{\Gamma}_{\rm T}\}$ $(\{\bar{f}_{\rm obj,P_3},\bar{\Gamma}_{\rm T}\})$ denotes the set of the objective value and SCNR value of problem  (P3) under the solutions in $\hat{\mathcal{V}}_{\rm P_3}$ ($\bar{\mathcal{V}}_{\rm P_3}$), respectively.  It can be readily inferred from \eqref{contra1} that the solutions in $\hat{\mathcal{V}}_{\rm P_3}$ are feasible and yield a higher objective value than those in $\bar{\mathcal{V}}_{\rm P_3}$, which contradicts the assumption that $\bar{\mathcal{V}}_{\rm P_3}$ is the optimal solution to problem (P3). This completes the proof.

\clearpage

\includepdf[pages=1-]{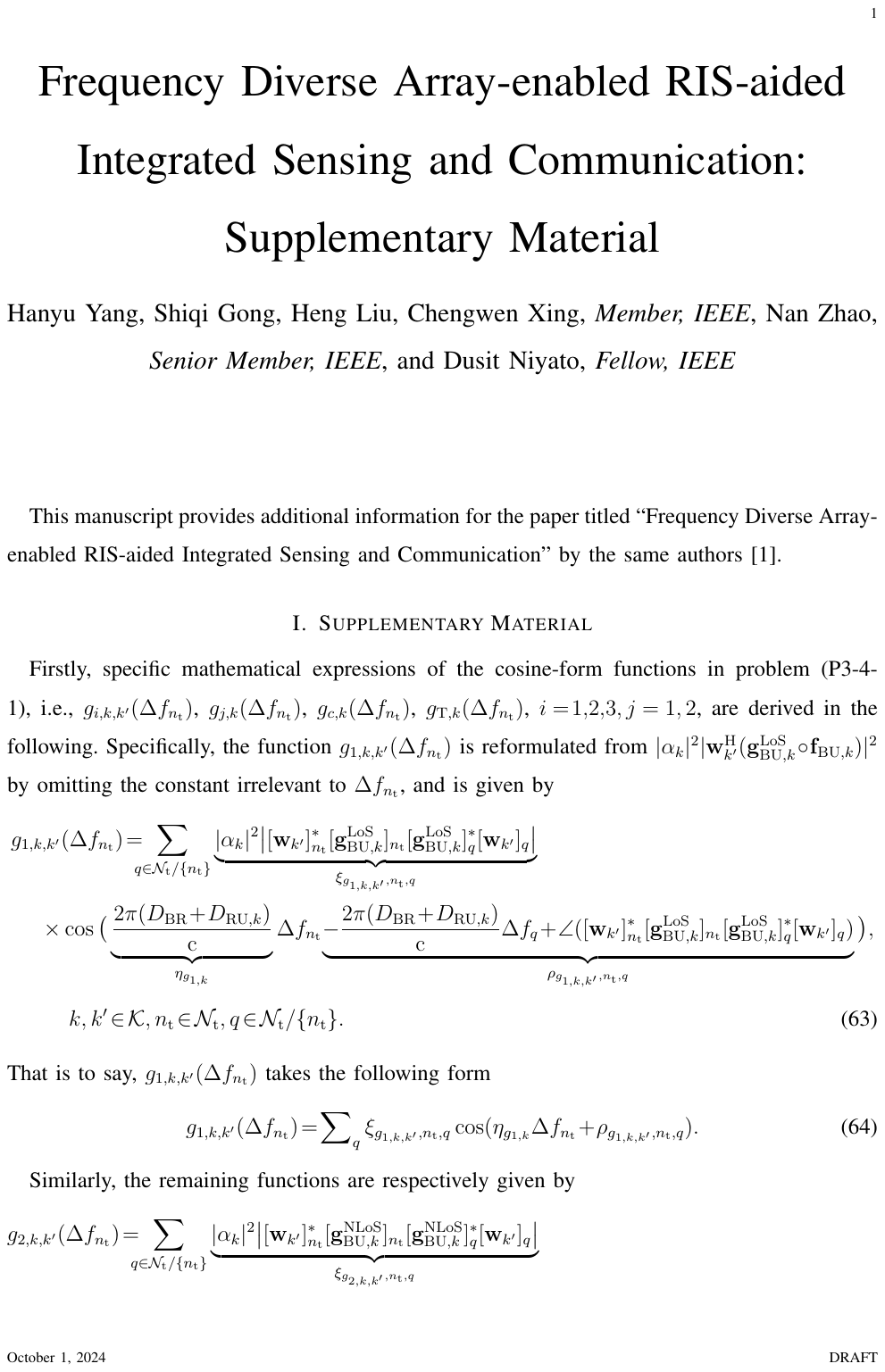}

\clearpage


\begin{thebibliography}{99}





%2
\bibitem{Intro_ref2}
R. Liu, M. Li, H. Luo, Q. Liu, and A. L. Swindlehurst, ``Integrated sensing and communication with reconfigurable intelligent surfaces: Opportunities, applications, and future directions,'' {\em IEEE Wireless Commun.}, vol.~30, no.~1, pp. 50--57, Feb. 2023.





        







\bibitem{rismec}
Z. Chen, J. Tang, M. Wen, Z. Li, X. Y. Zhang, and K.-K. Wong,``Reconfigurable intelligent surface assisted MEC offloading in NOMA-enabled IoT networks,'' {\em IEEE Trans. Commun.}, vol.7, no.8, pp. 4896-4908, May. 2023.

    %19
    \bibitem{Intro_ref4}
    Y. Guo, Y. Liu, Q. Wu, X. Li, and Q. Shi, ``Joint beamforming and power allocation for RIS aided full-duplex integrated sensing and uplink communication system,'' {\em IEEE Trans. Wireless Commun.}, vol. 23, no. 5, pp. 4627--4642, May. 2024.




    %3
    \bibitem{SUST1}
    X. Meng, F. Liu, S. Lu, S. P. Chepuri, and C. Masouros, ``RIS-assisted integrated sensing and communications: A subspace rotation approach: Invited paper,'' {\em Proc. IEEE Radar Conf. (RadarConf)}, San Antonio, TX, USA, pp. 1--6, May. 2023
    


    
%4
\bibitem{SUST2}
S. Yan et al., ``A reconfigurable intelligent surface aided 
dual-function radar and communication system,'' {\em Proc. IEEE Int'l. Symp. Joint Commun. Sensing (JC\&S)}, Seefeld, 
Austria, Mar. 2022. 


%5
\bibitem{MUST1}
J. Chen, K. Wu, J. Niu, Y. Li, P. Xu, and J. Andrew Zhang, ``Spectral and energy efficient waveform design for RIS-assisted ISAC,'' {\em IEEE Trans. Commun.}, Early Access.


%6
\bibitem{MUST2}
H. Zhang, ``Joint waveform and phase shift design for RIS-assisted integrated sensing and communication based on mutual information,'' {\em IEEE Commun. Lett.,} vol. 26, no. 10, pp. 2317--2321, Oct. 2022.



    %22
    \bibitem{RIS_ISAC3}
    X. Zhao, H. Liu, S. Gong, X. Ju, C. Xing and N. Zhao, ``Dual-functional MIMO beamforming optimization for RIS-aided integrated sensing and communication'', {\em IEEE Trans. Commun.}, vol. 72, no. 9, pp. 5411--5427, Sept. 2024.



%8
\bibitem{MUST4}
X. Song, X. Qin, J. Xu, and R. Zhang, ``Cramér-Rao bound minimization for IRS-enabled multiuser integrated sensing and communications,'' {\em IEEE Trans. Wireless Commun.}, vol. 23, no. 8, pp. 9714--9729, Aug. 2024.


\bibitem{MUST_gain}
J. Zuo, Y. Liu, C. Zhu, Y. Zou, D. Zhang, and N. Al-Dhahir, ``Exploiting NOMA and RIS in integrated sensing and communication'' {\em IEEE Trans. Veh. Technol.}, vol. 72, no. 10, pp. 12941--12955, Oct. 2023.




%9 
\bibitem{MUST5}
R. Liu, M. Li, Y. Liu, Q. Wu, and Q. Liu, ``Joint transmit waveform and passive beamforming design for RIS-aided DFRC systems,'' {\em IEEE J. Sel. Topics Signal Process.}, vol. 16, no. 5, pp. 995--1010, Aug. 2022.




    %14
    \bibitem{ISAC_ref4}
    C. Ouyang, Y. Liu and H. Yang, ``MIMO-ISAC: Performance analysis and rate region characterization,'' {\em IEEE Wireless Commun. Lett.}, vol. 12, no. 4, pp. 669--673, Apr. 2023.






    %40
    \bibitem{SCNR_ref}
    G. Cui, H. Li and M. Rangaswamy, ``MIMO radar waveform design with constant modulus and similarity constraints,'' {\em IEEE Trans. Signal Process.}, vol. 62, no. 2, pp. 343--353, Jan. 2014.

%24
\bibitem{Intro_ref7}
L. Lan, M. Rosamilia, A. Aubry, A. De Maio, and G. Liao, ``FDA-MIMO transmitter and receiver optimization,'' {\em IEEE Trans. Signal Process.}, vol. 72, pp. 1576--1589, Feb. 2024.





%32
\bibitem{FDA_ISAC5}
P. Gong, K. Xu, Y. Wu, J. Zhang, and H. C. So, ``Optimization of LPI-FDA-MIMO radar and MIMO communication for spectrum coexistence,'' {\em IEEE Wireless Commun. Lett.}, vol. 12, no. 6, pp. 1076--1080, Jun. 2023.






%28
\bibitem{FDA_ISAC1}
S. Ji, H. Chen, Q. Hu, Y. Pan, and H. Shao, ``A dual-function radar-communication system using FDA,'' in {\em Proc. IEEE Radar Conf.,} Oklahoma City, OK, USA, Jun. 2018, pp. 0224--0229.




     %31
     \bibitem{FDA_ISAC4}
     J. Jian, Q. Huang, B. Huang, and W.-Q. Wang, ``FDA-MIMO-based integrated sensing and communication system with frequency offset permutation index modulation,'' 2023, {\em arXiv:2312.14468}.

%29
\bibitem{FDA_ISA2}
S. Y. Nusenu, S. Huaizong, P. Ye, W. Xuehan, and A. Basit, ``Dual-function radar-communication system design via sidelobe manipulation based on FDA Butler matrix,'' {\em IEEE Antennas Wireless Propag. Lett.}, vol. 18, no. 3, pp. 452--456, Mar. 2019.

    
    %30
    \bibitem{FDA_ISAC3}
    H. Wu, B. Jin, Z. Xu, and X. Zhu, ``Transmit beamforming for dual-function radar-communication system based on FDA,'' {\em Int. Conf. on Signal Process. Syst. (ICSPS)}, Jiangsu, China, 2022, pp. 37--42.

   

%10
\bibitem{FDAandRIS}
C. Li, Stefan R. and Aydin S., ``Enhancing the secrecy rate with direction-range focusing with FDA and RIS'', 2024, {\em arxiv:2401.15154}.
    

      



%34
\bibitem{BR_LOS1}
J. Zhang et al., ``Enhancing performance of integrated sensing and communication via joint optimization of hybrid and passive reconfigurable intelligent surfaces,'' {\em IEEE Internet Things J.}, vol. 11, no. 19, pp. 32041--32054, Oct., 2024.




%27
\bibitem{FDA_LPF1}
Y. Zhang, Y. Zhang, J. Wang, S. Xiao, and W. Tang, ``Distance-angle beamforming for covert communications via frequency diverse array: Toward two-dimensional covertness,'' \emph{IEEE Trans. Wireless Commun.,} vol.~22, no.~12, pp. 8559--8574, Dec. 2023.







%37
\bibitem{FDA_LPF4}
Y. Liao, H. Tang, W.-Q. Wang, and M. Xing, ``A low sidelobe deceptive jamming suppression beamforming method with a frequency diverse array,'' {\em IEEE Trans. Antennas Propag.}, vol.~70, no.~6, pp. 4884--4889, Jun. 2022.








    %41
    \bibitem{Samefre}
    J. Jian, W.-Q. Wang, H. Chen, and B. Huang, ``Physical-layer security for multi-user communications with frequency diverse array-based directional modulation," {\em IEEE Trans. Veh. Technol.}, vol. 72, no. 8, pp. 10133--10145, Aug. 2023.




    



    %42
    \bibitem{QCQP_cons1}
    S.~Boyd and L.~Vandenberghe, \emph{Convex Optimization}. Cambridge, U.K.: Cambridge Univ. Press, 2004.

%43
\bibitem{Rayquo}
M. V. Pattabhiraman, ``The Generalized Rayleigh Quotient,'' {\em Can. Math. Bulletin}, vol.~17, no.~2, pp. 251--256, 1974.




    %44
    \bibitem{symADMM}
    B. He, F. Ma and X. Yuan, ``Convergence study on the symmetric version of ADMM with larger step sizes,'' {\em SIAM J. on Imag. Sci.,} vol.~9, no.~3, pp. 1467--1501, 2016.




    %45
    \bibitem{SubMethod}
    S. Boyd, L. Xiao, and A. Mutapcic, ``Subgradient methods,'' {\rm Lecture Notes of EE364b, Stanford University, Spring Quarter,}
    pp. 4--9, 2014.

    \bibitem{todo}
    H. Yang, S. Gong, H. Liu, C. Xing, N. Zhao, and D. Niyato, ``Frequency diverse array-enabled RIS-aided  integrated sensing and communication: Supplementary material'', {\em arxiv: TODO}, 2024.



    %46
    \bibitem{sim_RIS_SDR}
    Q. Wu and R. Zhang, ``Intelligent reflecting surface enhanced wireless network via joint active and passive beamforming,'' {\em IEEE Trans. Wireless Commun.}, vol.~18, no.~11, pp. 5394--5409, Nov. 2019.

%47
    \bibitem{rank-SDR}
    Y. Huang and D. P. Palomar, ``Rank-constrained separable semidefinite programming with applications to optimal beamforming,'' {\em IEEE Trans. Signal Process.}, vol.~58, no.~2, pp. 664--678, Feb. 2010.









    
    
    
    
    



  












    %\bibitem{FDA_LPF3}
    %Y. Xu, A. Wang and J. Xu, ``Range-angle transceiver beamforming based on semicircular-FDA scheme,'' \emph{IEEE Trans. Aerosp. Electron. Syst.,} vol.~58, no.~2, pp. 834--843, Apr. 2022.

    %被引
    %Range-Selective Beamforming With Reduced Computation to Illuminate/Suppress Multiple Dot-Like Loci in the Far Field, TAES

    %还有一篇tvt Physical-Layer_Security_for_Multi-User_Communications_With_Frequency 也用了这个方法









    \end{thebibliography}
\end{document}